%% file: main.tex
\documentclass[11pt]{article}
\usepackage{fullpage}

\usepackage[algo2e,linesnumbered,ruled,noend,noline]{algorithm2e}
\usepackage{amsthm,amsmath,amssymb}
\usepackage{graphicx}
\usepackage{tabularx,booktabs}
\usepackage{tablefootnote}
\usepackage{paralist}
\usepackage{array}
\usepackage{bm}
\usepackage{caption}
\usepackage{subcaption}
\usepackage{float}
\usepackage{cite}

\newcolumntype{C}[1]{>{\centering\let\newline\\\arraybackslash\hspace{0pt}}m{#1}}

\theoremstyle{plain}
\newtheorem{theorem}{Theorem}[section]
\newtheorem{claim}{Claim}

\newtheorem{lemma}[theorem]{Lemma}

\newtheorem{property}{Property}
\newtheorem{corollary}[theorem]{Corollary}

\newtheorem{definition}{Definition}

\newcommand{\sett}[1]{ \left\{ #1 \right\} }

\DeclareMathOperator{\poly}{poly}

\DeclareMathOperator{\polylog}{\mathrm{polylog}}
\DeclareMathOperator{\dist}{dist}

\DeclareMathOperator*{\argmin}{arg\,min}

\def\modelname {Heterogeneous MPC}

\newcommand{\N}{\mathbb{N}}

\newcommand{\degout}{\deg_{\mathrm{out}}}

\newcommand{\Boruvka}{\text{Bor{\r{u}}vka}}
\newcommand{\hrulealg}[0]{\vspace{1mm} \hrule \vspace{1mm}}

\SetKwInOut{Input}{Input}
\SetKwInOut{Output}{Output}
\SetKwBlock{SmallMachines}{small machines do in parallel:}{end}
\SetKwBlock{LargeMachine}{large machine:}{end}
\SetKwBlock{LocallyCompute}{locally compute:}{end}
\SetKw{Disseminate}{disseminate}
\SetKw{Aggregate}{aggregate}
\SetKwFor{DoParallel}{for}{do in parallel}{end}
\SetKwFor{DoForEachParallel}{foreach}{do in parallel}{end}
\SetKwFor{ForPack}{for}{do}{end}

\begin{document}

\title{Massively Parallel Computation in a Heterogeneous Regime}

\author{Orr Fischer \thanks{Computer Science Department, Weizmann Institute of Science, Israel. Email: orr.fischer@weizmann.ac.il} ~ Adi Horowitz \thanks{Blavatnik School of Computer Science, Tel-Aviv University, Israel. Email: adihorowitz1@mail.tau.ac.il}  ~   Rotem Oshman \thanks{Blavatnik School of Computer Science, Tel-Aviv University, Israel. Email: roshman@mail.tau.ac.il}}

\date{}

\maketitle

\abstract{Massively-parallel graph algorithms have received extensive attention over the past decade, with research focusing on three memory regimes: the superlinear regime, the near-linear regime, and the sublinear regime. The sublinear regime is the most desirable in practice, but conditional hardness results point towards its limitations.

In this work we study a \emph{heterogeneous} model, where the memory of the machines varies in size. We focus mostly on the heterogeneous setting created by adding a single near-linear machine to the sublinear MPC regime, and show that even a single large machine suffices to circumvent most of the conditional hardness results for the sublinear regime: for graphs with $n$ vertices and $m$ edges, we give
\begin{inparaenum}[(a)]
\item 
an MST algorithm that runs in $O(\log\log(m/n))$ rounds;
\item an algorithm that constructs an $O(k)$-spanner of size $O(n^{1+1/k})$ in $O(1)$ rounds; and
\item a maximal-matching algorithm that runs in $O(\sqrt{\log(m/n)}\log\log(m/n))$ rounds.
\end{inparaenum}
We also observe that the best known near-linear MPC algorithms for several other graph problems which are
conjectured to be hard in the sublinear regime
(minimum cut, maximal independent set, and vertex coloring) can easily be transformed to work in the heterogeneous MPC model with a single near-linear machine, while retaining their original round complexity in the near-linear regime.
If the large machine is allowed to have \emph{superlinear} memory, all of the problems above can be solved in $O(1)$ rounds.}

\input{Introduction}

\input{Prelinimaries}

\input{MST}

\input{Spanners}

\input{More_Results}
\input{Conclusion}

\section*{Acknowledgements}
Research funded by the Israel Science Foundation, Grant No. 2801/20, and also supported by Len Blavatnik and the Blavatnik Family foundation. This version is partially funded by the European Research Council (ERC) under the
European Union’s Horizon 2020 research and innovation programme (grant agreement No. 949083).

\bibliographystyle{plain}
\bibliography{mybibliography}

\appendix
\input{app_spanners}

\input{app_MST}
\input{app_spanner2.tex}
\input{app_additional_results}

\end{document}

%% file: Introduction.tex
\section{Introduction}\label{sec:intro}
% \TODO{delete} The widely believed $1-vs.-2$-CYCLE conjecture lemma states that in the sublinear regime of MPC, distinguishing between a cycle on $n$ vertices and two cycles on $n/2$ vertices requires $\Omega(\log(n))$ rounds. It was shown in \cite{hardness} that conditioned on this conjecture, lower bounds for other graph problems can be obtained as well. This includes $\Omega(\log\log(n))$ lower bound for Maximal Independent Set, constant approximation of Maximum Matching or constant approximation for Vertex Cover.

% In the near-linear regime of MPC, on the other hand, connectivity can be solved in $O(1)$. However, having each of the machines with a linear memory space is costly. We notice that in order to solve the connectivity problem efficiently (in constant rounds), we are only required to have a single machine with a linear memory space, using the linear sketches technique. This led us to the following question: what else could be solved efficiently in the sublinear MPC setting where we are given an extra one linear memory space machine?

The massively-parallel computation (MPC) model was introduced in~\cite{mapreduce} as a theoretical model of the popular Map-Reduce framework and other types of large-scale parallel computation, and has received significant attention from the distributed computing and algorithms communities.
In the MPC model we have $K$ machines, each with a local memory of size $s$,
operating on an input of size $N$ that is initially distributed arbitrarily across the machines. It is typically (but not always) assumed that $K \cdot s = \tilde{O}(N)$,
so that the total amount of memory in the system is of the same order as the input size.
All machines can communicate directly with one another;
the computation proceeds in synchronous rounds, with each machine sending and receiving at most $s$ bits in total per round.
Of particular interest in this model are graph problems, where the input is a graph $G = (V,E)$ on $n$ vertices, with the edges of the graph initially distributed arbitrarily among the machines.

The research on MPC graph algorithms typically considers three memory regimes:
the \emph{superlinear} regime, where the memory of each machine is
$s = O(n^{1+\gamma})$ for some $\gamma \in (0,1)$;
the \emph{near-linear} regime, where $s = O(n\polylog{n})$;
and the \emph{sublinear} regime, where $s = O(n^{\gamma})$ for some $\gamma \in (0,1)$.
Many graph algorithms have been developed for the various regimes (e.g., \cite{mapreduce,sorting,filtering,ANOY14,BKS17,coloring_in_sublinear,conn_log_diam,conn_log_diam_2,Sparsifying,MaximalMatching,efficient_coloring,GLM19,KPP20,GJN20,GGJ20,CFGUZ19,MIS,hardness,unweightewMincut,weightewMincut,CDP21,spanners,spanners2,MISinSublinear,CLMMOS20,NS19,CDP21a}).
%\Onote{I've cited here some results which are not algorithms; also, I'm not sure I have here a superlinear regime citation}.
The sublinear regime is the most desirable in practice:
when designing a distributed server cluster,
it is most economical to deploy many weak servers.
However, this configuration is also the most challenging, since it does not allow a single machine to store information about all the vertices of the graph.
It is conjectured that even some simple problems on sparse graphs are somewhat hard in the sublinear MPC model:
the ``2-vs-1 cycle'' problem, where we are asked to distinguish between a graph that is a single large cycle and graphs that comprise two cycles, is conjectured to require $\Omega(\log{n})$ rounds, and based on this conjecture and a beautiful connection to local distributed algorithms, \cite{hardness} establishes several conditional hardness results for sublinear MPC,
and these are refined in~\cite{CDP21a} (see, e.g.,~\cite{NS19} for an overview of more hardness results in the sublinear MPC model).

In this paper our goal is to interpolate between the various MPC regimes,
and ask what can be done in a \emph{heterogeneous} MPC regime where we have a small number of machines with large memories, and many machines with small memories.
Compared to the larger-memory regimes (near-linear and superlinear), the heterogeneous regime is more practical; thus, our hope is to have the best of both worlds --- 
the efficiency of the larger-memory regimes,
together with
the practicality and economic feasibility of the sublinear regime,
with the addition of a few strong (near-linear or superlinear) servers.

%\Anote{here} 
%\paragraph*{The \modelname{} model from a practical perspective.}
%Our model has an interesting practical perspective. In the design of a distributed server cluster, it is plausibly highly economical and resource efficient to include a single strong server alongside many weak servers compared to including many strong servers. As our work illustrates, the single strong server network configuration could have a power comparable or equivalent to the setting of many strong servers, in a variety of problems. \Onote{This is my attempt to explain a practical perspective of this work}

The starting point for our work is the observation that the ``$2$-vs-$1$ cycle'' problem
becomes trivial if we have even a single machine with memory $\Omega(n\log{n})$.
This motivates us to ask whether the problems whose conditional hardness rests on the hardness of the ``2-vs-1 cycle'' problem --- connectivity, minimum-weight spanning tree, maximal matching, and others --- also become easy given a small number of large machines.
For the sake of concreteness, since one large machine suffices to beat the ``2-vs-1 cycle'' problem, we focus throughout on a model where we have a single machine with memory $O(n\polylog{n})$, and the other machines have sublinear memory, $O(n^{\gamma} \polylog{n})$ for some $\gamma \in (0,1)$.
Unless specified explicitly differently, by saying \modelname{} model we refer to this setting that consists of a single large machine.
We show that indeed, for several of the problems from \cite{hardness}, a single large machine suffices to circumvent the conditional hardness results from the sublinear regime,
and either match or come close to the complexity of these problems in the near-linear regime.

\paragraph*{Our results}

Table~\ref{tab:results_summary} below summarizes our results:
for a graph on $n$ vertices and $m$
edges, with maximum degree $\Delta$ and diameter $D$, 
we compare the best known algorithms for the sublinear regime, the near-linear regime, and our results for the heterogeneous regime.

Although the heterogeneous model has not been explicitly studied in the past (to our knowledge),
several algorithms that were developed for the near-linear regime easily translate to the heterogeneous regime, as they either
already require only one large machine or can be easily modified to do so,
and we indicate those in the table.
%A description of these algorithms, and adaptations for the \modelname{} model (if necessary), is included in the appendix.
We highlight in bold the three problems --- MST, spanners, and maximal matching --- for which we developed substantially different algorithms for the heterogeneous regime.

%introduces the results in the \modelname{} model with respect to existing results in the near-linear and sublinear space MPC models. Table \ref{tab:super_linear_space} shows the results that can be obtained by considering a slightly super-linear space in a single machine in the \modelname{} model.
\begin{table}[H]\centering
\begin{minipage}{\textwidth}
\renewcommand*{\thempfootnote}{\arabic{mpfootnote}}
\begin{tabular}{|C{0.25\textwidth}|C{0.22\textwidth}|C{0.22\textwidth}|C{0.22\textwidth}|}
%  \begin{tabular}{|c|C{0.22\textwidth}|C{0.2\textwidth}|C{0.2\textwidth}|}
 \hline
 Problem & Sublinear MPC & Heterogeneous MPC &  Near-linear MPC
 \\ [0.5ex]
 \hline\hline
 Connectivity & $O(\log{D}+\log\log{n})$ \footnote{Here, $D$ is the diameter of a minimum spanning forest of the input graph.} \cite{conn_log_diam_2} 
 &
 $O(1)$ \cite{linearMeas}
 &
 $O(1)$ \cite{linearMeas}
 \\ \hline
 MST
 &
 $O(\log{n})$  \cite{conn_log_diam}
 &
 \boldmath{$O(\log\log(\frac{m}{n}))$}
 \textbf{[new]} 
 &
$O(1)$ \cite{linearMeas}
 \\
 \hline  
  $(1+\epsilon)$-approx. MST 
 &
  no better result known for $(1+\epsilon)$-approx.\ than for exact
 &
 $O(1)$ \cite{linearMeas} 
 &
exact in $O(1)$ \cite{linearMeas}
 \\ \hline
 $O(k)$-spanner%
 \footnote{For unweighted graphs.
 To obtain a spanner for weighted graphs,
 one can use the reduction from~\cite{spanners}.}
 of size $O(n^{1+1/k})$\footnote{We remark that for the \modelname{} model, since we can construct an $O(\log{n})$-spanner of size $\tilde{O}(n)$ in $O(1)$ rounds, 
we can also compute an $O(\log{n})$-approximation to all-pairs shortest paths (APSP), by storing the spanner on the large machine.}
 &
 $O(\log{k})$ \footnote{The algorithm finds an $O(k^{\log{3}})$-spanner of size $O(n^{1+1/k}\log{k})$.} \cite{spanners2}
 &
 \boldmath{$O(1)$}
 \textbf{[new]}
 &
  $O(1)$ \cite{spanners}
 \\ \hline
  Exact unweighted min-cut
  &
  $O(\polylog{n})$ \footnote{This algorithm was developed for PRAM, but it also works in sublinear MPC.} \cite{min_cut_PRAM}
  &
  $O(1)$ \cite{unweightewMincut}
  &
  $O(1)$ \cite{unweightewMincut}
  \\ \hline
 Approx.\ weighted min-cut
 &
 $(2+\epsilon)$-approx. in $O(\log{n} \cdot \log\log{n})$ \cite{weightewMincut}
 &
 $(1 \pm \epsilon)$-approx. in $O(1)$ \cite{weightewMincut} 
 &
exact in $O(1)$ \cite{weightewMincut}
 \\ \hline
  $(\Delta+1)$ vertex coloring
  &
  $O(\log\log\log{n})$ \cite{coloring_in_sublinear}
  &
  $O(1)$ \cite{coloring}
  &
  $O(1)$ \cite{coloring}
  \\ \hline 
 Maximal independent set
 &
 $O(\sqrt{\log\Delta}\cdot \log\log\Delta + \sqrt{\log\log{n}})$  \cite{Sparsifying}
 &
     $O(\log\log(\Delta))$ \cite{MIS}
  &
 $O(\log\log(\Delta))$ \cite{MIS}
 \\  \hline 
  Maximal matching
  &
  $O(\sqrt{\log\Delta}\cdot \log\log\Delta + \sqrt{\log\log{n}})$  \cite{Sparsifying}
  &
   \boldmath{
  $O(\sqrt{\log\frac{m}{n}}\log\log\frac{m}{n})$}
  \textbf{[new]} 
  &
$O(\log\log(\Delta))$ \cite{MaximalMatching} 
  \\  
 \hline
 \end{tabular}
 \caption{\label{tab:results_summary}Algorithms for fundamental graph problems in  sublinear MPC, \modelname{} and near-linear MPC}
\end{minipage}
\end{table}

These results can be further improved if the single large machine has more memory:
with memory
$\tilde{O}(n^{1+f(n)})$
for a function $f(n) = \Omega(1/\log{n})$,
we can solve MST in $O(\log(\log(m/n)/(f(n)\log{n})))$ rounds,
and the maximal matching algorithm of~\cite{filtering} can be used
to
find a maximal matching in $O(1/f(n))$ rounds.
However, to go all the way down to $O(1)$ rounds, we require $f = \Omega(1)$, that is,
the large machine needs to have $\Omega(n^{1+\gamma})$ memory for some constant $\gamma \in (0,1)$. It remains an intriguing open problem whether this can be improved.

\paragraph*{Our techniques}
In order to utilize the single large machine at our disposal, our algorithms for the \modelname{} model have the following flavor:
first, we select a \emph{sparse subgraph} $G'$ of our input graph $G$,
such that $G'$ has only $\tilde{O}(n)$ edges; the subgraph $G'$ is either sampled
at random or obtained by working to sparsify the input graph, or both.
We send $G'$ to the large machine
and have it compute some \emph{partial solution} for $G$.
Since the large machine sees only the subgraph $G'$,
the partial solution it finds typically induces an ``over-approximation''
of the true solution for $G$: for example, in the case of spanners,
the partial solution found by the large machine induces a spanner with more
edges than necessary.
However, we prove that the induced solution is ``not too large'' (or, in the case of MST, use the KKT sampling lemma~\cite{KKT}, which asserts exactly this).
The large machine then encodes the partial solution
in the form of a \emph{labeling}, $L : V' \rightarrow \sett{0,1}^{\polylog(n)}$,
of the vertices of $G'$, which it sends to the small machines:
each small machine is given the labels of all vertices whose edges it stores.
Finally, the small machines use the labels to select edges 
that are part of the true solution for $G$,
and we somehow combine the selected edges to come up with 
the correct solution for $G$, either by using the large machine again
or by other means.

Similar approaches of sub-sampling and over-approximation were previously used to obtain many efficient algorithms for a wide range of graph problems in different models, such as the streaming model (\cite{linearMeas, linear_sket_stream_1}), the Congested Clique model (\cite{connectivity_in_cc,random_sampling_cc,MSTinCC,JN18}), the near-linear MPC model (\cite{weightewMincut, unweightewMincut, coloring}), and the super-linear MPC model (\cite{filtering}).
Although not explicitly stated in these terms,
some of these results in fact require only a single large machine, or can be adapted to do so (see Table~\ref{tab:results_summary}).

% This approach of ``sparsification and overcompensation'' appears in several prior works on MPC and related models for several problems. For example, the KKT sampling lemma~\cite{KKT} and some techniques related to this lemma which we use in our MST algorithm have been previously used several algorithms of the distributed Congested Clique model \cite{connectivity_in_cc,random_sampling_cc,MSTinCC,JN18}, which is a model strongly related to the near-linear MPC model, and in the super-linear MPC model this approach has been used for finding a maximal matching in $O(1)$ rounds \cite{filtering}, though in a manner different than in our algorithm. \Onote{Is this an ok paragraph?}

\paragraph{Hardness results of sublinear MPC} 
The connection between lower bounds in the sublinear MPC model and Boolean circuit complexity was investigated first in \cite{RVW18}, and later on by \cite{circuit_simulation}. In \cite{RVW18} it is shown that proving a super-constant lower bound in sublinear MPC for any problem in $\mathsf{P}$ would imply $\mathsf{NC}_1 \subsetneq \mathsf{P}$, which suggests that any such lower bound is beyond the reach of current techniques. This has led the research community to focus more on conditional lower bounds (e.g. \cite{hardness,NS19,CDP21a}).
As previously discussed, the ``1-vs-2 cycle'' problem is conjectured to have round complexity
$\Omega(\log{n})$ in sublinear MPC;
this immediately implies the same lower bound for problems such as MST and shortest-paths.
A general framework for lifting existing lower bounds from the distributed $\mathsf{LOCAL}$ model to obtain conditional lower bounds in sublinear MPC was given in~\cite{hardness}, and refined and extended in~\cite{CDP21a}.%
\footnote{The framework from~\cite{hardness,CDP21a} applies only to algorithm satisfying a natural property called \emph{component stability}, and this limitation is somewhat inherent, as~\cite{CDP21a} showed that in fact component-unstable algorithms can break the conditional lower bounds from~\cite{hardness}.
The algorithms we present in this paper can trivially be made component-stable, because we can first solve connectivity on the large machine, and then work on each connected component separately but in parallel.}
%
%: this framework applies to a
%natural class of algorithms called \emph{component-stable algorithms}.
%The framework of~\cite{hardness} was later refined in~\cite{CDP21a}, and ~\cite{CPD21a} showed how to extend the original framework to lift both deterministic lower bounds and lower bounds with dependency on the maximum degree
%. In addition, \cite{CDP21a} demonstrated the limitation of ``component stability'' by providing component-unstable algorithms that break some of these conditional lower bounds, .
The conditional lower bounds of \cite{hardness,CDP21a} include an $\Omega(\log\log{n})$ conditional lower bound for any component-stable algorithm computing a constant approximation of maximum matching, a constant approximation of vertex cover, or a maximum independent set, an $\Omega(\sqrt{\log\log{n}})$ conditional lower bound for computing $(\Delta+1)$-coloring, and an $\Omega(\log{k})$ conditional lower bound on computing an $O(k)$-spanner with $O(n^{1+1/k})$ edges.

\paragraph{Additional related work}

Beyond the results summarized in Table~\ref{tab:results_summary},
there are some algorithms that can achieve good dependence on 
graph parameters other than $n$ and $\Delta$.
Of some relevance, for graphs with arboricity $\alpha$,
in the sublinear regime,
a maximal independent set 
or a maximal matching can be computed in
      $O(\sqrt{\log\alpha}\log\log\alpha + \log\log{n} \log\log\Delta)$
      rounds~\cite{MISinSublinear}.
  The arboricity $\alpha$ of a graph satisfies $m/n \leq \alpha \leq \Delta$,
  so our result for the \modelname{} is quantitatively better
  (of course,~\cite{MISinSublinear} works in the sublinear regime, so no direct comparison is possible).
  
%  Another approach that extends the MPC model into a more realistic one is the \emph{Adaptive MPC model}~\cite{AMPC}, where
%  the machines are allowed to make a bounded number
%  of adaptive read queries
%  to a shared read-only memory within each round.
%  This is intended to capture remote direct memory access (RDMA),
%  and it enables very efficient solutions to several of the problems we study here;
%  see~\cite{AMPC} for more details.
 %\Anote{I commented out the previous paragraph about AMPC, and inserted the one from the appendix}
Another approach that extends the MPC model into a more realistic one is the \emph{Adaptive MPC model}~\cite{AMPC}.
In the AMPC model, for an input of size $N$ there are $m$ machines, each with memory of size $S$, such that $m \cdot S = \Tilde{O}(N)$. Additionally, there is a collection of distributed data stores, denoted by $D_0, D_1, ...$, such that in the $i$'th round, each machine can read data from $D_{i-1}$ and
write to $D_i$. The number of bits a machine can read and write in one round is limited by $O(S)$. As opposed to the native MPC, it is not required for a machine to read all of the $O(S)$ bits at once in the beginning of the round. Instead, each machine can query bits based on previous queries that it has already made in the same round. This adaptive behavior was the reasoning for the model name. It was claimed by \cite{AMPC} that the AMPC model could actually be more practical than the general MPC in many cases. Specifically, they pointed out that the MPC model does not take into consideration the ability to perform RDMA operations, by which a data stored on a remote machine can be read with only a few microsecond latency without requiring a synchronized round of communication. The AMPC model tries to close this gap between theory and practice. Indeed, this model efficiently solves a few problems which were considered hard for the native MPC model. Among them are the minimum spanning tree problem, which can be solved in $O(\log\log_{\frac{m}{n}}(n))$ rounds, and maximal independent set, which can be solved in $O(1)$ rounds.

%% file: Prelinimaries.tex
\section{Preliminaries}\label{sec:prelims}

\paragraph*{Graph notation}
In this work we consider undirected graphs $G = (V,E)$,
which may be either weighted or unweighted.
We let $n$ denote the number of vertices in the graph,
and $m$ the number of edges.
If the graph is weighted, the weight function is given by $w : V \rightarrow \{1,\dots,\poly(n)\}$ (positive integer weights bounded by some fixed polynomial of $n$),
and we make the standard assumption that all edge weights are unique, and can each be represented in $O(\log{n})$ bits.

For a set $S \subseteq V$ of nodes,
let $E[S, V \setminus S] \subseteq E$ denote the edges that have one endpoint
in $S$ and one in $V \setminus S$.
Let $\dist_G(u,v)$ denote the unweighted (resp.\ weighted) distance between $u,v \in V$ in the unweighted (resp.\ weighted) graph $G$.
We omit the subscript $G$ when the graph is clear from the context.
 We use the words ``node'' and ``vertex'' interchangeably.

\paragraph*{The \modelname{} model}

In the \modelname{} model, 
we have one \emph{large machine}, which has memory $O(n \polylog(n))$,
and $K = m/n^{\gamma}$ \emph{small machines},
each with memory $O(n^{\gamma} \polylog(n))$, where $\gamma \in (0,1)$
is a parameter of the model.
The small machines are numbered $1,\ldots,K$, 
and when we say \emph{predecessor}, \emph{successor}, or \emph{consecutive machines}, we refer to
this numbering.
We assume that each machine has a source of private randomness (no shared randomness is assumed).

The input to the computation is a weighted or
unweighted undirected graph $G = (V,E)$.
The vertices $V$ of the input graph are fixed in advance,
and the edges $E$ are initially stored on the small machines
arbitrarily.
If $G$ is weighted, then the weight function
is specified by representing each edge as $(\sett{u,v}, w(\sett{u,v}))$.

The computation proceeds in synchronous rounds, where each machine can communicate with all the other machines, subject to each machine $M$ sending and receiving in total only as many bits as it can store in its memory (i.e., no more than $O(n \polylog{n})$ bits for the large machine,
or $\tilde{O}(n^{\gamma})$ bits for the small machines).
Between rounds, the machines may perform arbitrary computations on their local data, unbounded by time or by space.
At the end of the computation, the output is either stored on the large machine, if the size of the output allows, or distributed across the small machines if the output is too large.
 
 Throughout the paper, "with high probability" means with probability $1 - 1/n^d$ for some $d \geq 1$.

\paragraph*{Algorithmic tools}
In this section we describe some general tools that can be useful in \modelname{} model and we will use throughout the paper. They rely mostly on known results for the sublinear MPC regime.

\begin{claim} [Sorting \cite{sorting}] \label{sorting_claim}
     Given a set of $N$ comparable elements as input
     stored on the small machines, there is an
     $O(1)$-round algorithm that sorts the items on the small machines,
     such that at the end,
     for any two small machines $M < M'$,
     each item stored on $M$
     is no greater than any item stored on $M'$.
\end{claim}

One useful set of tools from prior work is fast aggregation or dissemination of information in sublinear MPC.
We use minor variants of these techniques that also may involve the large machine.

\begin{definition}
We say that $f : 2^S \rightarrow S$ is an
\emph{aggregation function}
if for all $k \geq 1$
and subsets $X_1,\ldots,X_k \subseteq S$
we have
$f( \sett{ f(X_1), \ldots, f(X_k) } ) = f( X_1 \cup \ldots \cup X_k)$.
%an aggregation function $f: 2^S \rightarrow S$ is a function %for which for all $k$ and sets $A_1, ... ,A_k \in 2^S$ such %that $f(A_1) = a_1 \wedge ... \wedge f(A_k) = a_k$, it holds %that $f(A_1 \cup ... \cup A_k) = f(\{a_1,...,a_k\})$.
\label{def:agg}
\end{definition}

Let $\N^S$ denote the set of all multisets whose elements
are drawn from $S$.
\begin{claim} [Aggregation \cite{efficient_coloring,distance_sketches}]\label{aggregation_claim}
Fix a domain $S$,
partitioned into predetermined subsets
$S_1,\ldots,S_t$, known in advance to all machines,
and let
$f : \N^S \rightarrow S$
 be an aggregation function that is fixed in advance 
 and known to all machines.
 We assume that the elements of $S$
 can be represented in $O(\polylog{n})$ bits.
Let $A \subseteq S$ be a set or a multiset
stored on the small machines, and let $A_i = A \cap S_i$.
Then there is a constant-round algorithm that computes $f(A_1),\ldots,f(A_t)$,
such that
at the end of the algorithm, for each $i$,
the value $f(A_i)$ is known to some small machine $M_i$
that is fixed in advance.
%Assume the input is a collection of sets $A_1,...,A_t$ over a domain $S$, that is $A_i \subseteq S$ for all $1 \leq i \leq t$. The size of an element in $S$ is $O(\log(n))$ bits. For each $1 \leq i \leq t$, there is an aggregation function $f_i: 2^{S} \rightarrow S$ which is known to the small machines. There is an algorithm that computes $f_i(A_i)$ for all $1 \leq i \leq t$ in constant number of rounds, using the small machines. 
\end{claim}

% \begin{claim} [aggregation tree \cite{efficient_coloring,distance_sketches}] \label{aggregation_claim}
%      Assume a collection of sets $A_1,...,A_n$ that contain in total $N$ comparable elements is given as input. The elements are stored in the sublinear machines in lexicographical sorted order, primarily sorted by their set index. An aggregation tree $T_i$ for a set $A_i$ is a constant-depth tree in which each node represents a machine. The leafs of the tree are all the machines that hold (part of) the elements in the set $A_i$. The degree of each inner-node is at most $n^{1-\epsilon}$.
%      There is a constant time algorithm for the sublinear machines to compute the tree $T_i$ for each $1 \leq i \leq n$. In the end of the algorithm, each machine that has a node associated with it in some aggregation tree $T_i$ knows the machine that is associated with its parent node in $T_i$.
% \end{claim}

\begin{claim}[Dissemination \cite{efficient_coloring,distance_sketches}]\label{spreading_claim}
Fix a domain $S$, partitioned into predetermined subsets $S_1,\ldots,S_t$ which are
known to all machines,
and let $A \subseteq S$ be stored on the small machines.
Let $A_i = A \cap S_i$.
Suppose the large machine holds values $x_1, \ldots, x_t \in \sett{0,1}^{(\polylog{n})}$.
There is a constant-round algorithm that disseminates
$x_i$ to every small machine that stores an element from $A_i$,
in parallel for all $i = 1,\ldots,t$.
%the input is a collection of sets $A_1,...,A_t$ over a domain $S$, that is $A_i \in 2^{S}$ for all $1 \leq i \leq t$. Assume the large machine holds a value $f_i(A_i)$ of size $O(\log(n))$ bits, for each $A_i$. There is an algorithm that disseminates $f_i(A_i)$ for all $1 \leq i \leq t$ in constant number of rounds, such that at the end, every small machine that stores an element from $A_i$ also stores $f_i(A_i)$.
\end{claim}

\begin{proof}[Proof of Claims~\ref{aggregation_claim} and \ref{spreading_claim}.]

For concreteness we focus on the context of graph problems and assume that for all $1 \leq i \leq t$, the set $A_i$ is associated with some node $u \in V$. These assumptions hold in all the places we use Claim~\ref{aggregation_claim} and \ref{spreading_claim} through out the paper, but we note that the claim can be proved without them.

First, for each edge $(u,v) \in E$ we make two directed copies, $(u,v)$ and $(v,u)$. We sort all directed edges on the small machines
in lexicographic order.
Let $M_u$ be the first machine that
stores some outgoing edges of vertex $u$.
Each small machine $M$
sends to its successor 
the largest vertex $u$
such that some edge $(u,v)$ is stored
on $M$.
If $M$ holds some edge $(u,v)$ 
such that the predecessor of $M$
holds no outgoing edges of $u$, then 
$M$ informs the large machine that $M$
is the first machine that holds an outgoing
edge of $u$ (i.e., $M_u = M$).
Let $M'_u$ be the last machine that holds
some outgoing edge of $u$.
The large machine now knows the range of machines $[M_u, M'_u]$ that hold edges adjacent to $u$, for each $u \in V$.
Let $T_u$ be a tree 
with branching factor $n^{\gamma}$,
rooted at $M_u$,
whose nodes are the machines
in the range $[M_u, M_u']$,
such that for any two machines $M < M'$,
the depth of $M$ (its distance from the root)
is no greater than the depth of $M'$.
There are at most $n^{1-\gamma}$ machines that store outgoing
edges of $u$,
so the depth of $T_u$ is $O((1-\gamma)/\gamma)$.
We fix $T_u$ in advance as a function of $u$, $M_u$, $M'_u$. The large machine informs $M_u$ about the range $[M_u, M'_u]$. Then, the machine $M_u$ knows all its children in $T_u$ and what is the size of all its children's sub-trees, so it sends to each of its children the size of its sub-tree. By induction, each machine $M$ in the range $[M_u,M'_u]$ that receives the size of its sub-tree from its parent in $T_u$, knows all its children in $T_u$ and what is the size of all its children's sub-trees. So, in $O((1-\gamma)/\gamma)$ rounds all machines in the range $[M_u, M'_u]$ can know the IDs of 
their parent and their children in $T_u$. 
Note that $M$ can only be an inner node in one tree $T_u$:
if $M$ is not a leaf in $T_u$, then there is some machine $M'$ in $T_u$
whose depth is greater than $M$'s depth, meaning that $M < M'$
and so $M$
is not the last machine that stores outgoing edges of $u$;
there can be at most one node $u$ such that $M$ stores
some but not all edges of $u$.
Thus, each machine participates in at most one tree.

For the proof of Claim~\ref{aggregation_claim}, we have each machine $M$ in the leaf of the tree $T_u$ computes $f(B_u)$ such that $B_u \subseteq A_u$ is the set of elements that $M$ holds and send the computed value to its parent. Each inner-node machine computes $f(Y)$ such that $Y$ is the set of computed values received from all its children. The machine $M_u$, which is the root of $T_u$ will have the final result $f(A_u)$ in $O((1-\gamma)/\gamma)$ rounds.

To prove Claim~\ref{spreading_claim},
we disseminate a value $x_u$ to the machines in $T_u$ by having each machine in $T_u$ sending $x_u$ to all its children, starting from the root $M_u$. All the machines in $T_u$ that hold an element from $A_u$ will know the value $x_u$ in $O((1-\gamma)/\gamma)$ rounds.
    
\end{proof}

Occasionally we will need to arrange the edges of a graph,
so that all edges adjacent to a given node are stored
on consecutive machines.
For this purpose we switch to a directed version of the graph,
where each edge $\sett{u,v}$ appears in both orientations,
$(u,v)$ and $(v,u)$. We then sort the edges by their source node:

\begin{claim}[Arranging Nodes]\label{consec_machines}
In $O(1)$ rounds, it is possible to arrange the edges of a directed graph $G = (V,E)$ on the small machines, such that 
\begin{enumerate}
    \item For each vertex $v \in V$, 
    the outgoing edges of $v$ are stored
    on consecutive small machines.
    \item For each $v \in V$,
    let $M_{\mathrm{first}}(v)$
    be the first small machine
    that holds some outgoing edge of $v$.
    Then the large machine knows $M_{\mathrm{first}}(v)$ for each $v \in v$,
    and each small machine $M$ knows whether $M = M_{\mathrm{first}}(v)$
    for each $v \in V$.
   \item For each $v \in V$, the large machine and the small machine $M_{\mathrm{first}}(v)$ know the out-degree $\degout(v)$ of $v$.
\end{enumerate}
\end{claim}

\begin{proof}[Proof of Claim~\ref{consec_machines}.]
    We sort all the edges using Claim~\ref{sorting_claim}.
    Next, using Claim~\ref{aggregation_claim},
    we compute the degree $\degout(v)$ of each node $v \in V$,
    using the
    domain $S = V \times \N$,
    the partition $\sett{ S_u = \sett{ u} \times \N }_{u \in V}$,
    the
    aggregation function $f(\sett{(u, x_1),\ldots,(u,x_t)}) = (u, \sum_{i = 1}^t x_i)$,
    and initial multiset $A = \sett{ (u,1) : (u,v) \in E, u < v }$.
    The small machines then report the node degrees
    to the large machines.
    Using this information,
    the large machine is able to identify the machine $M_{\mathrm{first}}(v)$ for each $v \in V$,
    and inform $M_{\mathrm{first}}$ itself.
\end{proof}

%The proof of these claims is deferred to Appendix~\ref{sec:app_prelims}.

% \Rnote{Not sure what it means that a node delivers messages in this lemma. Also, can't refer to $M_u$ which is defined inside another lemma's statement. Everything should be encapsulated.}
% \begin{claim}[neighbors messaging]\label{neigh_comm}
%      Assume that for each node $u \in V$, there is a message $m(u)$ of size $\polylog(n)$ bits which is known to at least one machine which holds (part of) $u$'s incident edges. There is a constant round algorithm that delivers $m(u)$ to all $u$'s neighbors for all $u \in V$, such that in the end of the algorithm, a machine that holds an edge $(u,v)$ knows the messages $m(u)$ and $m(v)$.
% \end{claim}
% \begin{proof}
%  Sort all edges using claim \ref{sorting_claim}. Define the set $A_u$ to contain all incident edges to node $u$, and the machine $M_u$ as the root of the aggregation tree given by Claim~\ref{aggregation_claim}. For each $u \in V$, $M_u$ broadcasts the message to all other machines that hold $u$'s edges using the aggregation tree. Then, for each edge $(u,v)$, attach the message $m$ to the edge label and switch the order between $u$ and $v$ to create a label $(v, u, m)$. Sort again all labels to deliver $u$'s message to the neighbor $v$.
% \end{proof}

%% file: MST.tex
\section{MST in $O(\log\log({m}/{n})$) Rounds}\label{MST_section}
\label{sec:MST}

In this section we show how to compute a minimum-weight spanning tree in
the \modelname{} model.
Although we are mostly interested in the case where the large machine
has near-linear memory,
we state a more general result, which can use a large machine
with memory $n^{1+f(n)}$ for any $f(n) = \Omega(1/\log{n})$:

%$O(\log ( \log_n(m/n) / f(n) )$ rounds,
%$O(\log (\frac{\log_n(\frac{m}{n})}{f(n)}))$ rounds,
%if we are given one machine with memory $\tilde{O}(n^{1+f(n)})$, and all the other machines are sublinear.

\begin{theorem}\label{MST_with_super_lin}
Given a single machine with memory size $\Tilde{\Omega}(n^{1+f(n)})$ and 
$\Omega(m/n^{\gamma})$ machines with memory $\tilde{\Omega}(n^{\gamma})$
(where $\gamma \in (0,1), f(n) = \Omega(1/\log{n})$),
there is an algorithm
that
computes a minimum spanning tree with high probability
in  $O(\log(\log (m/n)/(f(n)\log{n})) )$ rounds.
\end{theorem}

%The details for general $f$ can be found in Appendix~\ref{app:MST};
For most of the section, we focus on the case where 
the large machine has near-linear memory,
i.e.,
$f(n) = 1/\log{n}$.
In this case the round complexity we get is
$O(\log\log(m/n))$.  Finally, we show how the proposed algorithm can be generalized for any $f(n) = \Omega(1/\log{n})$ to prove Theorem~\ref{MST_with_super_lin}.

%\TODO{??? I think there was a switch from $f(n)$ to $1/f(n)$ at some point. Previously it said: ``By taking $f(n) = 1/\log n$, we get the following:''
%}\Anote{by taking $f(n) = 1/\log n$ we have $O(\log (\frac{\log_n(\frac{m}{n})}{f(n)})) = O(\log\log m/n)$ and $O(n^{1+1/ \log n} = O(n))$ memory, no?}
%\Anote{
%\begin{equation*}
%    \log (\frac{\log_n(\frac{m}{n})}{f(n)}) = \log \frac{(\log_2(\frac{m}{n}) / \log_2 n}{1/\log_2 n} = \log (\log_2 (\frac{m}{n}))
%\end{equation*}}

Before presenting our MST algorithm, we review the KKT sampling lemma, on which our algorithm relies.

\paragraph*{The KKT sampling lemma~\cite{KKT}}
The sampling lemma of Karger, Klein and Tarjan~\cite{KKT} was developed in the context of sequential MST computation, and has also been used in distributed MST algorithms, e.g., ~\cite{random_sampling_cc, MSTinCC, connectivity_in_cc,JN18}. 
Informally, the lemma asserts that for any graph $G$ whose MST we wish to compute, if we sample a random subgraph $H$ where each edge of $G$ is included with probability $p$, then a minimum-weight spanning forest (MSF) $F$ of $H$ can be used to dismiss from consideration all but $O(n/p)$ edges of $G$.

More formally, let $G = (V, E)$, let $H = (V, E')$ be a subgraph of $G$,
and let $F = (V, E'')$ be a spanning forest of $H$.
An edge $e \in E$ is called \emph{$F$-heavy} if adding $e$ to the forest $F$
creates a cycle, and $e$ is the heaviest edge in that cycle.
If $e$ is not $F$-heavy then we say that $e$ is \emph{$F$-light}.
No $F$-heavy edge can be in the MST of $G$; thus, we can restrict our attention to the $F$-light edges of $G$. The KKT sampling lemma states that the MSF of a \emph{random} subgraph $H$ of $G$ suffices to significantly sparsify $G$.

%\textbf{Random-Sampling Lemma.}
%An additional component of the MST algorithm is the random sampling lemma introduced by Karger, Klein and Tarjan \cite{KKT} in the context of sequential MST algorithms. First, they have defined the following:
%
%\begin{definition} [F-light edge \cite{KKT}]
%Consider a graph $G=(V,E)$ and a spanning forest $F=(V,E')$. An edge $e \in E$ is called $F$-light if by adding it to $F$, it creates a cycle in which it is not the heaviest edge. Otherwise, it is called $F$-heavy.
%\end{definition}
%
%Note that the edges of $F$ are all $F$-light. Also note, that for any forest $F$ , no $F$-heavy edge can be in the minimum spanning forest of $G$. The last observation directly derives from the cycle property. Based on this observation, it follows that if we find a spanning forest $F$ for which a large quantity of edges are $F$-heavy with respect to it, we actually succeed in sparsifying the problem significantly. It turns out that a minimal spanning forest $F$ computed over a random-sampled subgraph of $G$ is good enough for this purpose, as stated formally in the next lemma.
%
\begin{lemma}[Random-Sampling Lemma~\cite{KKT}]
Let $H$ be a subgraph obtained from $G$ by including each edge of $G$ independently
with probability $p$, and let $F$ be the minimum spanning forest of $H$. The expected number of $F$-light edges in $G$ is at most $n/p$, where $n$ is the number of vertices of $G$.
\label{lemma:KKT}
\end{lemma}

\paragraph*{Overview of our MST algorithm}
Our MST algorithm consists of two parts:
in the first part, we apply $O(\log\log(m/n))$ steps of the doubly-exponential $\Boruvka$ technique from~\cite{lotker}, where in the $i$-th step, each remaining vertex selects its $2^{2^i}$ lightest outgoing edges and merges with the vertices on the other side of those edges.
		After $O(\log\log(m/n))$ such steps, there remain at most $n' = n^2 / m$ vertices.
		This number is small enough that we can afford to switch to a different approach and use the KKT lemma.
		
	In the second part of the algorithm, we sub-sample the remaining graph, taking each edge with independent probability $p = n/m$, to obtain a subgraph with $O(n)$ edges.
		By the KKT sampling lemma, the expected number of $F$-light edges is at most $n'/p = (n^2 / m) / (n / m) = n$.
		Thus,
		we can now complete the MST computation by having the small machines identify which of the edges
		they store are $F$-light and send those edges to the large machine,
		which then computes the MST.
		However, the small machines cannot \emph{store} $F$, as it is too large;
		each machine needs to identify its $F$-light edges without knowing all of $F$.
		To that end, the large machine computes a labeling of the vertices of $F$,
		such that for each edge $\sett{u,v} \in E$, we can determine
		whether $\sett{u,v}$ is $F$-light from the labels of $u$ and $v$.
		We use the flow labeling scheme from \cite{KKKP04} for this purpose.
		The large machine sends to each small machine the labels of all vertices
		for which the small machine is responsible, and this allows the small machine to
		identify the $F$-light edges it stores and send them to the large machine.
%In high level, our MST algorithm can be divided into two parts. The first part is a vertex-reduction part, in which we apply $\log\log \frac{m}{n}$ iterations of exponentially speed-up Bor{\r u}vka algorithm in a similar manner to \cite{lotker} to obtain a contracted graph with at most $n' = \frac{n^2}{m}$ nodes. The second part is based on the random-sampling lemma of \cite{KKT}. We sample each edge independently with probability $p=\frac{n}{m}$ to obtain a random subgraph with $O(n)$ edges. Then, the linear machine computes an MST $F$ over this sampled subgraph. By the random-sampling lemma, there are at most $n'/p = \frac{n^2}{m} \cdot \frac{m}{n} = O(n)$ edges that are not cancelled due to the red rule with respect to the original graph and the MST $F$. These edges are identified by the sublinear machines using labels that the linear machine constructs. We note that a similar approach for this part was taken in \cite{MSTinCC} in the Node-Congested Clique model. These edges are then stored in the linear machine, which computes the final MST of the input graph.

Next, we explain how each part of the algorithm is carried out in the \modelname{} model.

\paragraph*{Doubly-exponential $\Boruvka$}\label{lotker_impl}
As we said above, the first part of our algorithm implements the doubly-exponential $\Boruvka$ technique,
which was first implemented by~\cite{lotker} in the Congested Clique.
The algorithm of~\cite{lotker} gradually contracts vertices, so that 
after $i$ steps we have at most $n/2^{2^i}$ remaining vertices:
in the $i$-th step, each vertex finds its $2^{2^i}$ lightest outgoing edges,
and merges with the vertices at the other side of those edges.
This reduces the number of vertices to at most $(n/2^{2^{i}})/2^{2^i} = n/2^{2^{i+1}}$.
Note that only $O(n)$ edges in total are identified as lightest edges in each step,
since at the beginning of the $i$-th step 
there are at most $n / 2^{2^i}$ vertices, and each selects at most $2^{2^i}$ edges.
After $O(\log\log{n})$ steps, only one vertex remains, and the MST then consists of all the
edges along which contractions were performed.

To implement doubly-exponential $\Boruvka$ in the \modelname{} model,
we use the small machines to find the lightest edges, and the large machine to perform vertex contractions
and store the MST edges identified along the way.
Only $O(\log\log(m/n))$ steps will be performed, so at the end, the graph is contracted down to $n^2 / m$ vertices instead of a single vertex.

At the beginning of step $i$, we represent the current graph by $G_i = (V_i, E_i, w_i)$,
where $V_i$ is a set of at most $n / 2^{2^i}$ contracted vertices,
each represented by an $O(\log{n})$-bit identifier,
and the edges $E_i$ are stored on the small machines together with their weights $w_i$.
To eventually find the MST of the original graph,
the large machine maintains a mapping $c_i : V \rightarrow V_i$ that associates each original vertex $v \in V$
with the vertex into which $v$ has been merged.
In addition,
together with each edge $\sett{u,v} \in E_i$ we also store the original graph edge $\sett{u', v'} \in E$
that $\sett{u,v}$ ``represents'':
the edge $\sett{u', v'}$ such that 
node $u'$ was merged into $u$, node $v'$ was merged into $v$, and $w_i(\sett{u,v}) = w(\sett{u', v'})$.
The original graph edge $\sett{u', v'}$ is sent together with $\sett{u, v}$ whenever $\sett{u,v}$ is sent anywhere. 

%The \emph{size}
%of a contracted vertex $u \in V_i$,
%denoted $s(u) = \sett{ v \in V : c(v) = u }$, is the number of original vertices merged to form $u$.

To perform step $i$, we first arrange the edges $E_i$ on the small machines, so that
the outgoing edges of each vertex $v \in V_i$ are stored on consecutive machines, sorted by their weight.
This involves creating two copies of each edge, as we need to consider both endpoints:
let $\tilde{E}_i = \sett{ (u,v) : \sett{u,v} \in E_i}$ be the set of all outgoing edges in $G_i$.
Using Claim~\ref{sorting_claim}, we sort $\tilde{E}_i$ across the small machines in order of the first vertex,
and then by edge weight (that is, $(u,v) < (u', v')$ if $u < u'$ or if $u = u'$ and $w(\sett{u,v}) < w(\sett{u', v'})$).

The large machine now collects,
for each vertex $v \in V_i$,
the $\min(2^{2^i},\deg_{\mathrm{out}}(v))$ lightest outgoing edges
of $v$.
This is done as follows:
\begin{itemize}
	\item Using Claim~\ref{consec_machines},
		the large machine learns the out-degree of each node $v \in V_i$,
		which small machines store outgoing edges of $v$.
	\item The large machine locally computes, for each $v \in V_i$
	and small machine $M$,
	the number $k(v,M)$ of edges among $v$'s lightest 
		$\min(2^{2^i}, \deg_{\mathrm{out}}(v))$ outgoing edges that are stored on
		machine $M$.
		(Recall that the edges are sorted by source node and then by weight, and the large machine knows how many edges are stored on each small machine, 
		so it can compute which machines hold the first $\min(2^{2^i}, \deg_{\mathrm{out}}(v))$ outgoing edges of $v$ and how many of those edges are stored on each machine.)
			\item For each $v \in V_i$,
		the large machine sends a query of the form $(v, k(v, M))$ to each small machine $M$ such that 
		$k(v, M) > 0$.
		(A single small machine may receive multiple queries,
		but no more than the number of nodes whose edges it stores.)

		We have $\vert V_i \vert \leq n / 2^{2^i}$,
		and at most $2^{2^i}$ machines $M$ have $k(v, M) > 0$;
		thus, the total number of bits sent by the large machine is $O(n \log{n})$.
	\item Each small machine $M$ that received a query of the form $(v, k)$ (possibly multiple queries per small machine) sends the lightest $k$ outgoing edges of $v$ that it stores to the large machine. The total number of bits received by the large machine is at most $(n / 2^{2^i}) \cdot 2^{2^i} \cdot \log{n} = O(n \log{n})$.
\end{itemize}

Finally, the large machine contracts the graph along all the edges it collected:
it examines the edges by weight, starting from the lightest,
and for each edge examined, it merges the vertices at the endpoints
of the edge into one vertex, re-names the vertices of all remaining (heavier) edges accordingly,
and discards heavier edges that have become internal (i.e., both their endpoints
are merged into the same vertex).
An $O(\log{n})$-bit unique identifier is assigned to the new vertex.%
\footnote{Identifiers may be re-used in different steps of doubly-exponential $\Boruvka$.}
The large machine also stores the set of original graph edges attached to all lightest edges that it used in some merging step;
these edges will be part of the MST output at the end.
%\TODO{Needs this?} \Anote{yes-the large machine doesn't have the original graph edges, so even if it keeps track of the mapping $c: V \rightarrow V_i$, it does not suffice for it to know which original edge should be added to the MST}
 %
%The large machine notifies the small machines about all the edges $(u,v) \in E_i$ that were used for the merging (each small machine is notified which of its edges were used).
%From the invariant, the small machines know which original edges $(u',v') \in E$ were used in the merging, and add those edges to the MST. 

Let $G_{i+1} = (V_{i+1}, E_{i+1})$ be the contracted graph computed
by the large machine,
and let $c'_i : V_i \rightarrow V_{i+1}$ map each vertex of $V_i$ to the vertex into which it was merged in $V_{i+1}$.
The large machine creates an updated contracted-vertex map, $c_{i+1} : V \rightarrow V_{i+1}$,
where $c_{i+1}(v) = c_i'(c_i(v))$.
Using Claim~\ref{spreading_claim}, the large machine disseminates the update map $c_i'$
to the small machines, so that in $O(1)$ rounds, every small machine that holds some edge adjacent
to a node $v \in V_i$ learns $c_i'(v)$.
Each small machine then updates the edges it stores: it discards any edge $\sett{u,v}$ that became internal
($c'_i(u) = c'_i(v)$), and re-names the vertices of the remaining edges according to $c'_i$ (preserving the weight of the edge and the original graph edge attached to it).
With the help of the large machine, the small machines also ensure that if parallel edges are created by the contraction,
then only the lightest edge between any two nodes is kept,
and the others are discarded. (This is easily done using a variant of Claim~\ref{aggregation_claim}.)
%Recall that we attach to each edge $(u, v) \in E_i$ an original graph edge $(u', v') \in E$
%such that $c_i(u') = u, c_i(v') = v$, and $w_i(u,v) = w(u',v')$;
%upon re-naming $u$ to $c_i'(u)$ and $v$ to $c_i'(v)$, we keep the same original graph edge attached,
%so that the new edge $(c_i'(u), c_i'(v)) \in E_{i+1}$ also has attached to it an original graph edge $(u', v') \in E$,
%with $c_{i+1}(u') = c_i'(c_i(u')) = c_i'(u), c_{i+1}(v') = c_i'(c_i(v')) = c_i'(v)$,
%and $w_{i+1}(u,v) = w_i(u', v') = w(u,v)$.

This first part of the algorithm ends after $\log\log(m/n)$ steps of doubly-exponential $\Boruvka$.
At that point, at most $n / 2^{2^{\log\log(m/n)}} = n^2 / m$ vertices remain.

\paragraph*{Sampling edges}
Let $G' = (V', E')$ be the contracted graph on $n^2 / m$ vertices computed in the first part.
In the second part of the algorithm, we randomly sample a subgraph $G_p = (V', E_p)$ of $G'$, 
where each edge is chosen independently with probability $p = n/m$.
The sampling is carried out by the small machines:
we arrange the edges $E$ on the small machines in some arbitrary order (without the duplication that was needed 
in the first part of the algorithm), and each machine selects each edge that it holds with independent probability $p$.
Since $p = n/m$,
with high probability we have $\vert E_p \vert = \tilde{O}( n )$, so the large machine can store the graph $G_p$.

\paragraph*{Identifying $F$-light edges}
After receiving the edges of $G_p$, the large machine computes an MSF $F$ of $G_p$.
Next we apply the flow labeling scheme from \cite{KKKP04},
which consists of a pair of algorithms:
a \emph{marker} algorithm $\mathcal{M}_{\mathrm{flow}}$ that takes $F$ and returns a vertex labeling $L : V' \rightarrow \sett{0,1}^{O(\log^2 n)}$,
and a \emph{decoder} algorithm $\mathcal{D}_{\mathrm{flow}}$ such that for any $u ,v \in V'$,
$\mathcal{D}_{\mathrm{flow}}( L(u), L(v) )$ returns the weight of the heaviest edge on the path between $u$ and $v$ in $F$.%
\footnote{In \cite{flow_labeling} the decoder algorithm returns the \emph{lightest} edge on the path,
but as pointed out in \cite{flow_labeling}, it is easy to modify the scheme so that the decoder
instead returns the \emph{heaviest} edge, and this is what we need here.}

To identify the $F$-light edges, the large machine applies $\mathcal{M}_{\mathrm{flow}}$ to $F$
to obtain the labeling $L : V' \rightarrow \sett{ 0,1}^{O(\log^2 n)}$.
By using Claim~\ref{spreading_claim}, the large machine then disseminates the labels $L(v)$ for each $v \in V'$, such that each small machine that holds an edge of $v$ knows the lable $L(v)$.
Finally, each small machine examines every edge $\sett{u,v}$ that it stores,
and discards $\sett{u,v}$ from memory iff $w(\sett{u,v}) > \mathcal{D}_{\mathrm{flow}}( L(u), L(v) )$.
The remaining edges are sent to the large machine, which computes and outputs an MST on all the edges it receives,
together with the edges of $G_p$.

By Lemma~\ref{lemma:KKT},
the expected number of $F$-light edges is at most $O( (n^2 / m) / (n / m) ) = O(n)$,
and by Markov,
for some constant $\alpha > 0$,
there are at most $\alpha \cdot n$ $F$-light edges with probability $1/2$.
We count the $F$-light edges by having each small machine send to the large machine the number of edges it selected, and if the total number of $F$-light edges is at most $\alpha n$, the small machines send them to the large machine; otherwise we abort.
To reach success probability $1 - 1/n^c$,
we repeat the entire process (sampling $G_p$, computing $F$, finding and counting the $F$-light edges) $O(\log{n})$ times, in parallel.

After receiving the $F$-light edges for some successful instance $G_p$, the large machine completes the MST computation by adding the $F$-light edges to $F$, and computing an MST $F'$ on the resulting graph.
Finally, the large machine outputs the MST of the original graph $G$: this consists of all the edges that were used to merge nodes during the doubly-exponential $\Boruvka$ phase, plus the edges of the MST $F'$ of $G_p$.

We conclude this section by analyzing the general case described in Theorem~\ref{MST_with_super_lin}, where the large machine may have superlinear memory:
\begin{proof}[Proof of Theorem~\ref{MST_with_super_lin}.]
%This is a generalization of the algorithm presented above that utilizes the extra memory of the large machine.
Applying $t$ steps of doubly-exponential Bor{\r u}vka algorithm results in a contracted graph with $n' = \frac{n}{n^{2^tf(n)}}$ vertices (since in the $i$'th step, there is enough space in the large machine to hold up to $n^{2^if(n)}$ outgoing edges from each component). For the random-sampling step, we fix the sampling probability at $p=\frac{1}{n^{2^tf(n)+f(n)}}$, so that the number of light edges will be at most $O(n'/p) = O(n^{1+f(n)})$ which fits the memory of the large machine.
We get that for all $t \geq  \log(\frac{\log_n(m/n)}{f(n)})$, it holds that
\begin{equation*}
\frac{\log_n(m/n)}{f(n)} - 2 \leq \frac{\log_n(m/n)}{f(n)} \leq 2^t.
\end{equation*}
Simplifying yields
\begin{equation*}
%\begin{split}
    % m \cdot p = \frac{m}{ n^{2^tf(n)+f(n)}} \leq \frac{m}{n^{\log_n(m/n) + f(n)}} = \frac{m}{(m/n) + n^{f(n)}} = 
   % & \frac{\log_n(m/n)}{f(n)} - 2 \leq \frac{\log_n(m/n)}{f(n)} \leq 2^t \\
   % & \log_n(m/n) - 2f(n) \leq 2^tf(n) \\
    %& \log_n(m/n) - f(n) \leq 2^tf(n) + f(n) \\
     %   & \frac{m}{n^{1+f(n)}} \leq n^{2^tf(n)+f(n)} %\\
            %&
        %    \frac{m}{ n^{2^tf(n)+f(n)}} \leq n^{1+f(n)} \\
           %&
           m \cdot p \leq n^{1+f(n)},
  %  \end{split}
\end{equation*}
and thus, after $O(\log(\frac{\log_n(m/n)}{f(n)}))$ iterations of doubly-exponential Bor{\r u}vka, we can apply the final step of random-sampling (which takes constant number of rounds) to find the MST.
\end{proof}

%% file: Spanners.tex
\section{$O(k)$-Spanner of Size $O(n^{1+1/k})$ in $O(1)$ Rounds}
\label{sec:spanners}

Recall that \emph{$k$-spanner} of a weighted or unweighted graph $G = (V,E)$ is a subgraph $H$ of $G$,
such that for every $u,v \in V$ we have 
$\dist_H(u,v) \leq k \cdot \dist_G(u,v)$ (with the distances weighted if $G$ is weighted,
or unweighted if $G$ is unweighted).
The \emph{size} of the spanner $H$ is the number of edges in $H$. It is known that for any graph $G$, there exists a $(2k-1)$-spanner of size $O(n^{1+1/k})$ (see e.g. \cite{AIDDJ93}), and this bound is believed to be tight, assuming Erd{\H{o}}s's Girth Conjecture \cite{E64}.
In this section we show that for unweighted graphs, it is possible to compute a $(6k-1)$-spanner of size $O(n^{1+1/k})$ in $O(1)$ rounds in the \modelname{} model.
For weighted graphs, a known reduction to the unweighted case (see, e.g.,~\cite{spanners})
yields a $(12k-1)$-spanner of size $O(n^{1+1/k} \log{n})$ in $O(1)$ rounds.

\begin{theorem}\label{spanner_the}
For every $k \leq \log{n}$, there is an $O(1)$-round algorithm for the \modelname{} model that 
computes an $O(k)$-spanner of expected size $O(n^{1+1/k})$ 
in the unweighted case, or $O(n^{1+1/k}\log{n})$ in the weighted case.
\end{theorem}
While the result is stated in terms of the expected size of the spanner, we can
of course get a spanner whose size is $O(n^{1+1/k})$ w.h.p.\ (or $O(n^{1+1/k}\log{n})$,
 in the weighted case) by repeating the algorithm $O(\log{n})$ times in parallel and taking the smallest spanner found.

By Theorem~\ref{spanner_the}, for $k=\log{n}$, we get a $O(\log{n})$-spanner of size $\Tilde{O}(n)$ in $O(1)$ rounds.
This spanner can fit in the memory of the large machine,
and it can be used to solve approximate all-pairs shortest paths (APSP):

\begin{corollary}
 There is an $O(1)$-round algorithm in \modelname{} that computes w.h.p.\ an $O(\log(n))$-multiplicative approximation to all-pairs shortest paths
 in weighted or unweighted graphs.
\end{corollary}

\paragraph*{Overview of our spanner algorithm}
We focus on unweighted graphs, since the weighted case can be solved by reduction
 to the unweighted case, as we said above.
 
 The algorithm consists of two main ingredients. The first is the clustering-graph method from~\cite{spanners}, which was developed in the context of the Congested Clique and near-linear MPC:
 it computes a collection of
 $O(\log\Delta)$ \emph{clustering graphs}, $A_0,\ldots,A_{\log\Delta - 1}$, where the $i$-th graph $A_i$ has at most $O(n \cdot i / 2^i)$ vertices and at most $O(n \cdot 2^i)$ edges,
 such that we can combine spanners for each of the clustering graphs into a (slightly worse) spanner for the original graph.
 Although it was originally developed for machines with near-linear or larger memory, this method from~\cite{spanners} is readily adapted to the \modelname{} model.
 
 The next step is to compute a spanner for each of the clustering graphs.
For this purpose,
we present a modified version of the well-known Baswana-Sen spanner~\cite{BS}, where instead of computing the spanner over the entire graph, we first sub-sample the edges to obtain a subgraph that can fit on the large machine.
The large machine operates on the subgraph,
and computes \emph{some} of the information needed to find a spanner.
It then sends this information to the small machines,
which locally select which edges to add to the spanner.
Because the large machine has access only to a subset of edges,
this results in an ``overapproximation'', with the small machines adding ``too many'' edges compared to the true Baswana-Sen spanner, but not too many ---  
we prove that for a graph on $r$ vertices, if we sample each edge independently with probability $p$, the resulting spanner will be
of size $O(k r^{1+1/k} / p)$.
This is larger by a factor of $1/p$ compared to the Baswana-Sen spanner, whose size is $O(k r^{1+1/k})$.
By setting the sampling probability $p$ appropriately for each of the clustering graphs, we are able to obtain spanners that can be combined into one $O(k)$-spanner of size $O(n^{1+1/k})$ for the original graph, in $O(1)$ rounds.

\paragraph*{The clustering graphs}

Following~\cite{spanners}, for a graph $G = (V,E)$ with maximum degree $\Delta$,
we construct $O(\log\Delta)$ \emph{clustering graphs}, $A_0,\ldots,A_{\log\Delta - 1}$,
with the following properties:
\begin{itemize}
    \item The graph $A_0$ has $n$ vertices and $O(n)$ edges,
    and for  each $1\leq i \leq \log\Delta -1$, the graph $A_i$ has $O(n i / 2^i)$ vertices and $O(n \cdot 2^i)$ edges.
    \item There exists a transformation
    that takes the clustering graphs $A_0,...,A_{\log\Delta-1}$ and a $(2k-1)$-spanner $H_i$ for each such $A_i$, 
    and yields a  $(6k-1)$-spanner 
    $H$ for $G$, of size $O(n) + \sum_{i = 0}^{\log\Delta - 1} \vert H_i \vert$.%
    \footnote{Essentially, one takes the union of the spanners of the individual
    clustering graphs, but replacing each edge $e$ of a clustering graph $A_i$
    with an edge of the original graph $G$ to which $e$ is associated.}
\end{itemize}

In~\cite{spanners} it is shown that such clustering graphs can be constructed in the Congested Clique and in near-linear MPC.
With minor adaptations, similar graphs can be constructed in $O(1)$
rounds in the \modelname{}
model. We also show that if the $(2k-1)$-spanners $H_0,\ldots,H_{\log\Delta - 1}$ of $A_0,\ldots,A_{\log\Delta - 1}$ are initially stored on the small machines, we can compute
the transformation that yields a $(6k-1)$-spanner of $G$ in $O(1)$ rounds in the \modelname{} model.
More details can be found in the appendix.

\paragraph*{Computing spanners for the clustering graphs}
To compute a spanner for each clustering graph $A_i$, we present a modified version of the well-known Baswana-Sen algorithm~\cite{BS}, which allows us to construct
a spanner in $O(1)$ rounds for each clustering graph.
We begin by reviewing the Baswana-Sen spanner, and then explain our modified algorithm.
See Fig.~\ref{fig:spanner} for an illustration of the first step taken by the two versions of the algorithm, to demonstrate the difference between them.

\paragraph*{The Baswana-Sen spanner~\cite{BS}}
The Baswana-Sen spanner of a graph $G = (V,E)$ is computed by first choosing $k$ \emph{center sets},
$V = C_0 \supseteq C_1 \supseteq \ldots \supseteq C_{k-1} \supseteq C_k = \emptyset$,
where for each $i = 1,\ldots,k-1$,
the set $C_i$ is obtained from $C_{i-1}$ by selecting each vertex of $C_{i-1}$ with probability $1/n^{1/k}$.
Thus, in expectation, $C_i$ is of size $O(n^{1-i/k})$.
Each vertex $c \in C_i$ will serve as the center of a \emph{level-$i$ cluster}, $X_i(c) \subseteq V$,
such that in the subgraph induced by $X_i(c)$,
the eccentricity of $c$ is at most $i$.
The level-$i$ clusters are vertex-disjoint, but do not necessarily cover all vertices in $V$.

During the algorithm, as we go through the levels $i = 0,\ldots,k$,
vertices may be moved from cluster to cluster;
we let $c_i(v)$ denote the center of $v$'s level-$i$ cluster,
that is, the center $c \in C_i$ such that $v \in X_i(c)$,
or $c_i(v) = \bot$ if there is no such center.
Initially, $c_0(v) = v$ for all $v \in V$.
If $c_{i-1}(v) \neq \bot$ but $c_{i-1}(v) \not \in C_i$ (in other words, if $v$ belonged to the cluster of center $c$ after step $i-1$, but $c$ is no longer a center in $C_i$), then we say that $v$ \emph{becomes unclustered} in step $i$.
In this case Baswana-Sen attempts to re-cluster $v$ by adding it to some adjacent cluster
whose center $c$ is ``still alive'', $c \in C_i$.
If there is no such adjacent cluster, node $v$ remains unclustered, and we add to the spanner
one edge from $v$
to each adjacent level-$(i-1)$ cluster;
when this occurs, we say that vertex $v$ is \emph{removed} at step $i$.

The pseudocode for the Baswana-Sen algorithm is given in Algorithm~\ref{alg:BS_0} below.
If $c_i(v)$ is not explicitly set by the algorithm, then $c_i(v) = \bot$.

% \begin{algorithm2e} [ht]
% \caption{BaswanaSen($G = (V,E)$,$k$)}\label{alg:BS_0}
% \begin{algorithmic}[1]
% \State{$H \gets \emptyset$, $C_0 \gets V$, $\forall v \in V: \  c_0(v) \gets v$}
% \For{$i = 1,\ldots,k$}
% 	\If{$i = k$}
% 		\State{$C_i \gets \emptyset$}
% 	\Else
% 		\State{$C_i \gets$ sample each $ c \in C_{i-1}$ w.p. $1/n^{1/k}$}
% 	\EndIf
% 	\ForEach{$v \in V$ with $c_{i-1}(v) \neq \bot$}
% 		\If{$c_{i-1}(v) \in C_i$}
% 			\State{$c_i(v) \gets c_{i-1}(v)$}
% 		\ElsIf{ $\exists u \in N(v)$ with $c_{i-1}(u) \in C_i$}
% 		\label{line:BS0_neighbor}
% 			\label{line:recluster0}
% 			\State{$c_{i}(v) \gets c_{i-1}(u)$}
% 			\State{$H \gets H \cup \sett{u,v}$}
% 			\label{line:recluster1}
% 		\Else
% 		\ForEach{$c \in C_{i-1}$ s.t. $\exists u \in N(v)$ s.t. $c_{i-1}(u) = c$}\label{line:uncluster0}
% 				\State{$H \gets H \cup \sett{ \sett{ u,v } }$}
% 			\EndFor
% 		\label{line:uncluster1}
% 		\EndIf
% 	\EndFor
% \EndFor
% \end{algorithmic}
% \end{algorithm2e}

\begin{algorithm2e} [ht]
\caption{BaswanaSen($G$,$k$)}\label{alg:BS_0}
\DontPrintSemicolon
\SetKwInOut{Input}{Input}
\SetKwInOut{Output}{Output}

\Input{$G$ is a weighted, undirected graph, $k$ is an integer}
\Output{$H$ is a $(2k-1)$-spanner of $G$ of expected size $O(kn^{1+1/k})$}
\hrulealg

$H \gets \emptyset$, $C_0 \gets V$, $\forall v \in V: \  c_0(v) \gets v$\;
\For{$i = 1,\ldots,k$} {
	\If{$i = k$} {
		$C_i \gets \emptyset$\;
	}
	\Else {
		$C_i \gets$ sample each $ c \in C_{i-1}$ w.p. $1/n^{1/k}$\;
	}
	\ForEach{$v \in V$ with $c_{i-1}(v) \neq \bot$} {
		\If{$c_{i-1}(v) \in C_i$} {
			$c_i(v) \gets c_{i-1}(v)$\;
		}
		\ElseIf{ $\exists u \in N(v)$ with $c_{i-1}(u) \in C_i$} {
		\label{line:BS0_neighbor}
			\label{line:recluster0}
			$c_{i}(v) \gets c_{i-1}(u)$\;
			$H \gets H \cup \sett{u,v}$\;
			\label{line:recluster1}
		}
		\Else {
		\ForEach{$c \in C_{i-1}$ s.t. $\exists u \in N(v)$ s.t. $c_{i-1}(u) = c$} {  \label{line:uncluster0}
				$H \gets H \cup \sett{ \sett{ u,v } }$\;
		}
		\label{line:uncluster1}
		}
	}
}
\end{algorithm2e}

%More formally, the spanner $H$ is initialized to the empty set.
%We then proceed in $k$ steps, $i = 1,\ldots,k$: in the $i$-th step, for every vertex $v$
%with $c_{i-1}(v) \neq \bot$,
%\begin{itemize}
	%\item If $c_{i-1}(v) \in C_i$ (that is, $v$'s level-$(i-1)$ cluster center
		%is still ``alive'' at level $i$), then we set $c_i(v) \leftarrow c_i(v)$.
	%\item Otherwise, if $v$ has some neighbor $u$ such that $c_{i-1}(u) \in C_i$,
		%we choose one such neighbor $u$, add the edge $\sett{u,v}$ to $H$,
		%and set $c_i(v) \leftarrow c_{i-1}(u)$.
	%\item Finally, if $c_{i-1}(v) \not \in C_i$
		%and $v$ has no neighbor $u$ such that $c_{i-1}(u) \in C_i$,
		%we set $c_i(v) \leftarrow \bot$.
		%Then, for every level-$i$ center $c \in C_i$,
		%if $v$ has some neighbor $u$ with $c_{i-1}(u) = c$,
		%we choose one such neighbor $u$ and add $\sett{u,v}$ to $H$.
		%(Note that we add only one edge per cluster to which $v$ is adjacent.)
%\end{itemize}

\paragraph*{Modified Baswana-Sen}
We describe a simple modification of Baswana-Sen that makes it suitable for implementation in the \modelname{} model,
at the cost of yielding a larger spanner;
in some sense, we ``over-approximate'' the set of edges that would be taken
for the Baswana-Sen spanner.
In each step $i$ of the algorithm, we sample a subgraph $G_i$ of the original graph $G$,
where every edge of $G$ is included with independent probability $p$.
Then, in line~\ref{line:BS0_neighbor} of Algorithm~\ref{alg:BS_0},
instead of examining all neighbors of $v$ in the original graph $G$,
we only consider neighbors of $v$ in the subgraph $G_i$;
that is, we replace ``$u \in N(v)$'' with ``$u \in N_i(v)$'', where $N_i(v)$
denotes the neighborhood of $v$ in $G_i$.
The rest of the algorithm remains the same. However, for didactic purposes,
we re-arrange the pseudocode, and move lines~\ref{line:uncluster0}--\ref{line:uncluster1}
of Algorithm~\ref{alg:BS_1}
to a separate loop at the end of the algorithm,
as these lines will be executed later in our implementation (and they will by the small machines
rather than the large machine).
The result is given in Algorithm~\ref{alg:BS_1} below, with the 
change from Algorithm~\ref{alg:BS_0} underlined for clarity.

% \begin{algorithm2e} [ht]
% \caption{ModifiedBaswanaSen($G = (V,E)$,$k$,$p$)}\label{alg:BS_1}
% \begin{algorithmic}[1]
% \State{$H \gets \emptyset$, $C_0 \gets V$, $\forall v \in V: \  c_0(v) \gets v$}
% \label{line:BS_start}
% \For{$i = 1,\ldots,k$}
% 	\State{$V_i \gets V$}
% 	\State{$E_i \gets$ sample each $e \in E$ w.p. $p$}
% 	\label{line:BS_sample_G_p}
% \EndFor
% \For{$i= 1,\ldots,k$}
% 	\If{$i = k$}
% 		\State{$C_i \gets \emptyset$}
% 	\Else
% 		\State{$C_i \gets$ sample each $ c \in C_{i-1}$ w.p. $1/n^{1/k}$}
% 	\EndIf
% 	\ForEach{$v \in V$ with $c_{i-1}(v) \neq \bot$}
% 		\If{$c_{i-1}(v) \in C_i$}
% 			\State{$c_i(v) \gets c_{i-1}(v)$}
% 		\ElsIf{\underline{$\exists u \in N_i(v)$} with $c_{i-1}(u) \in C_i$}
% 			\State{$c_{i}(v) \gets c_{i-1}(u)$}
% 			\State{$H \gets H \cup \sett{u,v}$}
% 		\EndIf
% 	\EndFor
% \EndFor
% \label{line:BS_mid}
% \ForEach{$v \in V$ and $i = 1,\ldots,k$ s.t. $c_{i-1}(v) \neq \bot \land c_i(v) = \bot$}
% \label{line:BS_mid1}
% 	\ForEach{$c \in C_{i-1}$ for which $\exists u \in N(v)$ s.t. $c_{i-1}(u) = c$}
% 		\State{$H \gets H \cup \sett{ \sett{ u,v } }$}
% 	\EndFor
% \EndFor
% \label{line:BS_end}
% \end{algorithmic}
% \end{algorithm2e}

\begin{algorithm2e} [ht]
\caption{ModifiedBaswanaSen($G = (V,E)$,$k$,$p$)}\label{alg:BS_1}
\DontPrintSemicolon
\SetKwInOut{Input}{Input}
\SetKwInOut{Output}{Output}

\Input{$G$ is a weighted, undirected graph, $k$ is an integer, $p$ is a probability}
\Output{$H$ is a $(2k-1)$-spanner of $G$ of expected size $O(kn^{1+1/k}/p)$}
\hrulealg

$H \gets \emptyset$, $C_0 \gets V$, $\forall v \in V: \  c_0(v) \gets v$\;
\label{line:BS_start}
\For{$i = 1,\ldots,k$} {
	$V_i \gets V$\;
	$E_i \gets$ sample each $e \in E$ w.p. $p$\;
	\label{line:BS_sample_G_p}
}
\For{$i= 1,\ldots,k$} {
	\If{$i = k$} {
		$C_i \gets \emptyset$\;
	} 
	\Else {
		$C_i \gets$ sample each $ c \in C_{i-1}$ w.p. $1/n^{1/k}$\;
	}
	\ForEach{$v \in V$ with $c_{i-1}(v) \neq \bot$} {
		\If{$c_{i-1}(v) \in C_i$} {
			$c_i(v) \gets c_{i-1}(v)$\;
	    }
		\ElseIf{\underline{$\exists u \in N_i(v)$} with $c_{i-1}(u) \in C_i$} {
			$c_{i}(v) \gets c_{i-1}(u)$\;
			$H \gets H \cup \sett{u,v}$\;
		}
	}
}
\label{line:BS_mid}
\ForEach{$v \in V$ and $i = 1,\ldots,k$ s.t. $c_{i-1}(v) \neq \bot \land c_i(v) = \bot$} {
\label{line:BS_mid1}
	\ForEach{$c \in C_{i-1}$ for which $\exists u \in N(v)$ s.t. $c_{i-1}(u) = c$} {
		$H \gets H \cup \sett{ \sett{ u,v } }$\;
	}
}
\label{line:BS_end}
\end{algorithm2e}

In our analysis of the modified Baswana-Sen algorithm we use the following simple claim:
\begin{claim}
	For every $x > 1$ and every $\ell \geq 1$
	we have $\ell (1 - 1/x)^{\ell} < x$.
	\label{claim:exp_bound}
\end{claim}
\begin{proof}[Proof.]
	Fix $x > 1$.
	Since $(1-1/x)^x < 1/e$ for all $x > 1$,
	we have
	\begin{equation*}
		\ell(1-1/x)^{\ell} 
		=
		\ell \left[ \left( 1-1/x \right)^x \right]^{\ell / x}
		<
		\ell e^{-\ell/x}.
	\end{equation*}
	We know that $z e^{-z} < 1$ for all $z \in \mathbb{R}$,
	and thus,
	$\ell e^{-\ell /  x} < x$,
	which completes the proof.
\end{proof}

\begin{lemma}
	The modified Baswana-Sen algorithm computes a $(2k-1)$-spanner of $G$,
	comprising $O(kn^{1+1/k}/p)$ edges in expectation.
	\label{lemma:modifiedBS}
\end{lemma}
\begin{proof}[Proof.]
	Let $H$ be the set of edges output by our modified Baswana-Sen.
	An easy induction on the number of steps shows that:
	\begin{itemize}
		\item The eccentricity of each center $c \in C_i$ in the subgraph
			induced by $X_i(c)$ on $H$ is at most $i$: this is because
			in each step $i$, if we set $c_i(v) = c$,
			then either $c_{i-1}(v) = c$ or $v$ has some neighbor $u$
			with $c_{i-1}(u) = c$, and in this case we add $\sett{u,v}$ to $H$.
		\item For each $v \in V$, if $c_i(v) \neq \bot$, then $c_j(v) \neq \bot$ for all $j < i$.
	\end{itemize}

	We first bound the stretch of $H$.
	Let $\sett{u,v} \in E$,
	let $i$ be the step where $v$ is removed,
	and assume w.l.o.g.\ that $u$ is removed
	no earlier than $v$,
	that is, $c_{i-1}(u) \neq \bot$.
	Let $c = c_{i-1}(u)$.
	In step $i$, after we set $c_i(v) = \bot$,
	we add to $H$ one edge connecting $v$ to each adjacent level-$(i-1)$ cluster;
	in particular, since $X_{i-1}(c)$ is adjacent to $v$ (as $u \in N(v)$),
	there is some $w \in X_{i-1}(c) \cap N(v)$ such that $\sett{v,w} \in H$.

	Since the eccentricity of $c$ in the subgraph induced by
	$X_{i-1}(c)$ on $H$ is at most $i-1 \leq k - 1$,
	there is a path of length at most $k - 1$ between $u$ and $c$ in $H$,
	and a path of length at most $k - 1$ between $c$ and $w$ in $H$.
	Together with the edge $\sett{u,w}$, we obtain a path of length at most $2k-1$ between $u$ and $v$ in $H$.

	Now let us bound the expected size of $H$.
	Edges are added to $H$ at the following points in the algorithm:
	\begin{itemize}
		\item Upon re-clustering a node:
		if $v \in V$ becomes unclustered at level $i$,
			but has a neighbor in $G_i$ that is still clustered,
			we add one edge to $H$.
			The total number of all such edges
			added in all $k$ steps is at most $k \cdot n$
			(strictly speaking, $(k - 1) \cdot n$,
			since in step $k$ we have $C_k = \emptyset$
			and all nodes are unclustered).
\item 	Upon removing a node:
if node $v \in V$ 
is removed
in step $i$, 
then we add one edge from $v$ to each level-$(i-1)$ cluster $X$ adjacent to $v$ (i.e.,
each cluster $X$ such that $X \cap N(v) \neq \emptyset$).
We show that the expected number of edges of this type that are added for a given node $v$ in a given step $i$ is bounded by $O(n^{1/k}/p)$; multiplying this by $n$ and by $k$ gives us a bound on the total number of edges added in all steps by all nodes, in expectation.

Fix a node $v$ and a step $i$, and fix all the random choices for the preceding steps.
Let $c_1,\ldots,c_{\ell}$ be the centers of all level-$(i-1)$-clusters adjacent to $v$,
not including $v$'s level-$(i-1)$ cluster (if any).
Let us say that the pair $(v, c_j)$ \emph{survives} in step $i$
if $c_j \in C_i$, and in addition, $G_i$ contains
at least one edge $\sett{v, u}$ such that $u \in X_i(c_i)$.

We consider three possible events:
\begin{itemize}
\item Node $v$ does not become unclustered in step $i$. In this case node $v$ is not removed, and we add no edges.
\item Node $v$ becomes unclustered in step $i$, but for some $1 \leq j \leq \ell$,
the pair $(v, c_j)$ survives.
In this case, we re-cluster $v$ instead of removing it, and no edges of the type we are currently bounding are added to the spanner.
\item Node $v$ becomes unclustered in step $i$, and for all $1 \leq j \leq \ell$
the pair $(v,c_j)$ does not survive.
In this case node $v$ is removed in step $i$, causing
us to add $\ell$ edges to the spanner.
\end{itemize}
For each $j$, the probability that $(v,c_j)$ survives is at least $1 - p/n^{1/k}$:
the probability that $c_j \in C_i$ is $1/n^{1/k}$, and the probability that some edge connecting $v$ to $X_i(c_j)$ is sampled into $G_i$ is at least $p$ (since we know that $X_i(c_j)$ is adjacent to $v$, i.e., $G$ contains at least one edge from $v$ to $X_i(c_j)$).
This holds independently for each $j$, and it is also independent of whether $v$
becomes unclustered or not, because we did not count $v$'s own level-$(i-1)$ cluster (if any). Thus, the probability that none of the pairs $(v,c_1),\ldots,(v,c_{\ell})$ survives is at most $(1-p/n^{1/k})^{\ell}$, even conditioned on $v$ becoming unclustered.
All together, we see that the expected number of edges added if $v$
is removed in step $i$ is
\begin{align*}
    \ell \cdot \left( 1 - \frac{p}{n^{1/k}}\right)^{\ell} \leq \frac{n^{1/k}}{p},
\end{align*}
where the last inequality uses the fact that
$\ell(1 - 1/x)^\ell < x$ for every $x > 1$ and $\ell \geq 1$.

%Let $N(v,C_{i-1})$ denote the set of level-$(i-1)$ cluster adjacent to $v$, and $N_i(v)$ the set of neighbors in the sampled graph $G_i$ (line~\ref{line:BS_sample_G_p} of the algorithm). The probability that $v$ becomes unclustered in this iteration equals the probability that none of the clusters in $N(v,C_{i-1})$ was both sampled into $C_i$ and has an adjacent edge to $v$ that was sampled into $N_i(v)$, which is $\leq (1-p/n^{1/k})^{|N(v,C_{i-1})|}$. Thus, the expected number of edges added in this iteration by a node $v$ is given by
%			\begin{equation*}
 %               |N(v,C_{i-1})| (1-\frac{p}{n^{1/k}})^{|N(v,C_{i-1})|} \leq \frac{n^{1/k}}{p}
			%\end{equation*}
%by claim \ref{claim:exp_bound}.
%Overall, the expected number of edges added by all nodes in $k - 1$ iterations is %$(k-1)n^{1+1/k}/p$.
		\item At level $k$, all nodes become unclustered,
			and we add an edge from each vertex $v$
			to each level-$(k-1)$ cluster adjacent to $v$.
			The expected number of level-$(k-1)$ clusters
			is at most $n / n^{(1/k)\cdot(k-1)} = n^{1/k}$,
			so the expected total number of edges added
			is at most $n \cdot n^{1/k} = n^{1+1/k}$.
	\end{itemize}
	All together, the expected size of $H$ is $O(k n^{1+1/k} / p)$.
\end{proof}

\begin{figure}
\centering
\begin{subfigure}[b]{1\linewidth}
\centering
   \includegraphics[scale=0.6]{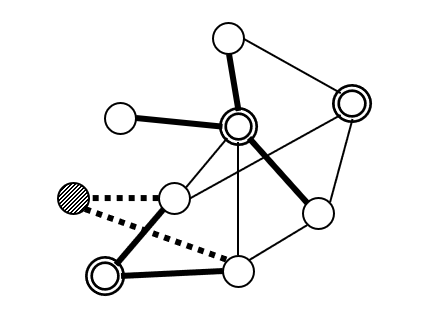}
   \caption{The first level of Baswana-Sen in the original graph.}
   \label{fig:Ng1} 
\end{subfigure}

\begin{subfigure}[b]{1\linewidth}
\centering
   \includegraphics[scale=0.6]{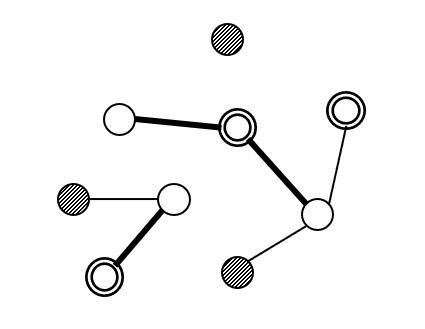}
   \caption{The first level of our modified Baswana-Sen, as executed by the large machine on the sub-sampled graph. The large machine is able to re-cluster fewer nodes than true Baswana-Sen (Fig.~\ref{fig:Ng1} above), and as a result, it adds fewer edges and removes more nodes.}
   \label{fig:Ng2}
\end{subfigure}

\begin{subfigure}[b]{1\linewidth}
\centering
   \includegraphics[scale=0.6]{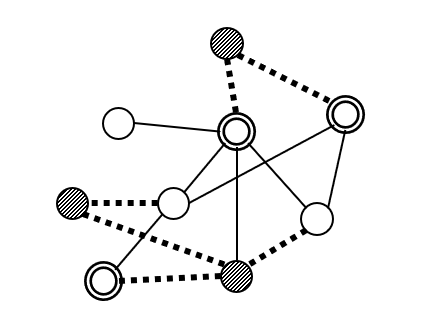}
   \caption{The first level of our modified Baswana-Sen, as executed on the small machines after learning the cluster to which each node is assigned by the large machine, and which nodes were removed. Since the large machine removes more nodes than true Baswana-Sen, the small machines add more edges than true Baswana-Sen.}
   \label{fig:Ng3}
\end{subfigure}

\caption[]{The first level of Baswana-Sen, the original version and our modified version. Initially, each node forms its own cluster ($C_0 = V$). Nodes selected for $C_1$ are indicated with double circles. Nodes that are no longer clustered ($C_0 \setminus C_1$) must either be re-clustered by adding them to some neighboring cluster (these nodes are shown as white circles), or removed (these nodes are shown as a striped circle). Bold lines indicate edges added to the spanner when a node is re-clustered, and dashed lines indicate edges added to the spanner when a node is removed.}
\label{fig:spanner}
\end{figure}
\paragraph*{Implementing modified Baswana-Sen in the \modelname{} model}
Assume that the input to the modified Baswana-Sen algorithm is a graph $G = (V,E)$ stored on the small machines,
an integer $1 \leq k \leq \log{n}$,
and a sampling probability $p \in (0,1)$
such that $p \cdot \vert E \vert = \tilde{O}(n)$.
To implement modified Baswana-Sen,
we have the small machines sample the subgraphs $G_1,\ldots,G_k$ locally
(sampling each edge indenendently with probability $p$),
and send them to the large machine.
This suffices for the large machine to carry out lines~\ref{line:BS_start}--\ref{line:BS_mid}
of Algorithm~\ref{alg:BS_1},
where we compute the clusters,
as this part of the algorithm depends only on $G_1,\ldots,G_k$
and does not require knowledge of the full graph $G$.

After executing lines~\ref{line:BS_start}--\ref{line:BS_mid},
the large machine sends to the small machines the clusters that it computed during the run,
specifying for each node $v$ and level $i$ the center $c_i(v)$ of the cluster to which $v$ belonged
in level $i$:
each small machine that stores some edge adjacent to $v$ receives the centers $(c_0(v),\ldots,c_k(v))$. This is done in constant number of rounds using Claim~\ref{spreading_claim}.

Finally, the small machines carry out lines~\ref{line:BS_mid1}--\ref{line:BS_end} of Algorithm~\ref{alg:BS_1}:
we add, for each vertex $v$ that was removed in some step $i$,
one edge connecting $v$ to each level-$(i-1)$-cluster that $v$ is adjacent to.
To find these edges, 
each small machine $M$ creates a set of \emph{candidates},
\begin{align*}
	A_M &= \sett{ (v,c,u) \in V^3 : \exists i \ c_{i-1}(v) \neq \bot, c_{i}(v) = \bot, 
	\right.
	\\
	&
	\qquad \left. c_{i-1}(u) = c \neq \bot, \text{ and } \sett{u,v} \in E }
\end{align*}
The candidate $(v,c,u)$ represents the fact node $v$ was removed
at level $i$, and at that point it was adjacent to a vertex $u$ belonging to a level-$(i-1)$ cluster
centered at $c$.
The edge $\sett{u,v}$ is thus a candidate for being added to the spanner, but we must ensure that we take
only one edge per level-$(i-1)$ cluster that $v$ is adjacent to:
using Claim~\ref{aggregation_claim},
in $O(1)$ rounds
we select for each vertex $v$ and center $c$
the smallest vertex $u$ such that $(v,c,u)$ is a candidate,
and add edge $\sett{u,v}$ to the spanner.

% To do this, we sort all the candidates $(u,c,v)$
% lexicographically.
% Then, each small machine sends to the small machine
% that follows it in the sort the largest candidate
% that it holds (in the lexicographic sort);
% this allows each small machine to learn,
% for each candidate $(u,c,v)$ that it holds, whether it is the smallest candidate for connecting $u$ to $c$.
% Finally, each machine $M$ selects
% the set of edges $\sett{u,v}$ such that
% $M$ holds a candidate of the form $(u,c,v)$,
% and this candidate is the smallest that connects $u$ to $c$.

\paragraph*{Putting everything together}
The variant of Baswana-Sen that we introduced above is particularly suitable for computing
spanners over the clustering graphs in the \modelname{} model. Recall that each of the clustering graphs $A_i$ contains at most $O(n \cdot i / 2^i)$ vertices and $O(n 2^i)$ edges, except for $A_0$, which has $O(n)$ vertices and $O(n)$ edges. 
Accordingly, we set our sampling probability to
$p_i = \min\sett{k^2i^{1+1/k}/2^i, 1}$.
Now there are three cases:
\begin{itemize}
    \item For $i = 0$, we can afford to send the entire edge set $E_0$ to the large machine, since $\vert E_0\vert = O(n)$.
    \item For $i > 0$ such that $p_i = 1$,
    we can also afford to send the entire edge set $E_i$ to the large machine:
    if $p_i = 1$,
    then $2^i \leq k^2 i^{1+1/k} \leq k^2 i^2$,
    and thus $\vert E_i\vert = O(n 2^i) = O(k^2 i^2 n)$.
    Since $k = O(\log{n})$ and $i \leq \log\Delta$,
    this means that $\vert E_i\vert = O(n \polylog{n})$.
    Thus, the large machine
    can compute an optimal spanner of size $O( (ni / 2^i)^{1+1/k} )$.
    
    %= O(\polylog n)$ and thus $|E_i| = O(n2^i) = O(n\polylog n)$. In this case, we simply send all edges $E_i$ to the large machine to compute a spanner $H_i$ with optimal parameters, that is $(2k-1)$-spanner of size $A_i = O((ni/2^i)^{1+1/k})$.
    \item For $i > 0$ such that $p_i = k^2i^{1+1/k}/2^i < 1$,
    we use modified Baswana-Sen.
    In this case 
    we have
    $p_i \cdot \vert E_i\vert = O( k^2 i^{1+1/k} / 2^i \cdot n 2^i) = O( k^2 i^{1+1/k} n ) = \tilde{O}(n)$,
    so indeed after sub-sampling the graph with probability $p_i$,
    the resulting edge set can fit on the large machine.
    By Lemma~\ref{lemma:modifiedBS},
    we obtain a $(2k-1)$-spanner $H_i$
    of expected size $O(k(ni/2^i)^{1+1/k} \cdot 2^i/k^2 i^{1+1/k}) = O(n^{1+1/k}/k2^{i/k})$.
\end{itemize}
Taking the maximum
of the options above, 
for each $i \geq 1$,
the spanner $H_i$ constructed for the $i$-th clustering graph has expected size
\begin{equation*}
    h_i = O\left( n^{1+1/k}
    \left(
    \left( \frac{i}{2^i}\right)^{1+1/k}
    +
    \frac{1}{k 2^{i/k}}
    \right)
    \right)
    .
\end{equation*}

Using the fact that
$\sum_{i = 0}^{\infty} (i^2 / 2^i) = 6$
and 
$\sum_{i = 0}^{\infty}
    1/(k 2^{i/k})
    <
    1/(1 - 1/2^{1/k})
    < 1$, we have that 
    $\sum_{i = 0}^{\infty} h_i = O(n^{1+1/k})$
    in expectation.
    
    Thus, the $(6k-1)$-spanner $H$ of $G$ 
    obtained by combining the spanners of the individual clustering graphs
    has an expected size of $O(n^{1+1/k})$.

%% file: More_Results.tex
\section{Maximal Matching in $O(\sqrt{\log(m/n)}\log\log(m/n))$ Rounds}
In this section we give an algorithm for maximal matching in the \modelname{} model, and show:

\begin{theorem}\label{MM_theorem}
There is an $O(\sqrt{\log{d}}\log\log{d})$-round algorithm in the \modelname{} model that computes a maximal matching with high probability in graphs of average degree $d$.
\end{theorem}

In our algorithm, we rely on the following claim from \cite{Sparsifying}:
\begin{lemma}[Section 3.6 and Proof of Theorem 3.2 in \cite{Sparsifying}]\label{lem:sub_MPC_MM}
In sublinear MPC, it is possible in $O(\sqrt{\log\Delta}\log\log{\Delta})$ rounds to find a matching $M$, such that the number of edges
with both endpoints unmatched in $M$ is at most $m/\Delta^{10} \leq n$.
\end{lemma}
After finding the matching guaranteed by Lemma~\ref{lem:sub_MPC_MM},
the number of edges that have both endpoints unmatched
is small enough to store all of them on the large machine and compute a maximal matching over them.
Thus
we obtain the following immediate corollary in \modelname{}:
\begin{corollary}\label{cor:sub_MPC_MM}
In \modelname{}, it is possible to find a maximal matching in $O(\sqrt{\log\Delta}\log\log{\Delta})$ rounds.
\end{corollary}

This is not yet our final result, since we want an algorithm
whose running time depends on the average degree $d$ rather than the maximum degree $\Delta$. Thus, we proceed in three phases:

\paragraph*{Phase 1:}
We divide the vertices into \emph{low-degree vertices},
$V_{\ell} = \sett{\ v : \deg(v) \leq d^2}$,
and \emph{high-degree vertices}, $V_h = V \setminus V_{\ell}$.
There are at most $n/d$ high-degree vertices:
by Markov, if we choose a random vertex $v$,
we have $\Pr[\deg(v) \geq d^2] \leq 1/d$,
and therefore $\vert V_h\vert \leq n/d$.

Using only the small machines, we apply the procedure of Corollary~\ref{cor:sub_MPC_MM}
to the graph induced by the low-degree vertices $V_{\ell}$, to obtain a maximal matching $M_1$. As the maximum degree in this graph is $d^2$, this takes $O(\sqrt{\log{d}}\log\log{d})$ rounds. The small machines send $M_1$ to the large machine. 

\paragraph*{Phase 2:}
For each vertex $v \in V_h$,
the large machine collects $2d \log{n}$ random incident edges of $v$
(including both neighbors in $V_{\ell}$
and in $V_h$), or all edges incident to a $v$ if $\deg(v) < 2 d\log{n}$. Denote this set by $E'(v)$. To do this, each small machine assigns a uniformly random rank $r(e) \in_U \{1,\dots,n^5\}$ to each edge $e$ that it stores.
With probability $1 - 1/n$ each edge is assigned a unique rank,
and we then select the $2d\log{n}$
lowest-ranked edges incident to each vertex (or all edges incident to the vertex)
and send
those edges to the large machine.
This is done in a manner 
similar to Section~\ref{MST_section} (in the MST algorithm, where the large machine
collects a fixed number of the lightest outgoing edges of each vertex).
Note that since $\vert V_h\vert \leq n/d$, the total number of edges
collected by the large machine is $O(n/d \cdot d\log{n}) = O(n\log{n})$.

The large machine greedily constructs a matching $M_2$, as follows:
initially, $M_2 = \emptyset$.
The large machine orders the vertices $V_h$ arbitrarily, and goes over the vertices in this order.
For each vertex $u$ examined, if $u$ is still unmatched in $M_1 \cup M_2$
and there is an edge $\sett{u,v} \in E'(u)$ such that $v$ is also unmatched in $M_1 \cup M_2$, then the large machine chooses one such neighbor $v$,
and adds $\sett{u,v}$ to $M_2$.

After constructing $M_2$, the large machine informs
the small machines about vertices that are matched in $M_1 \cup M_2$:
for each vertex $v$ and small machine $M$
that stores some edge adjacent to $v$,
the large machine informs $M$
whether or not 
$v$ is matched in $M_1 \cup M_2$.
This is done using Claim~\ref{spreading_claim}.

%, using the procedure of Claim~\Onote{Add claim from prelims}, and send them to the large machine.
%The large machine locally computes a matching $M_2$ over the received edges, as described next: First, the large machine chooses a random order over the vertices $V_h$. It goes over the vertices by this order and for each unmatched vertex, choose an arbitrary edge to an unmatched neighbor (unmatched both in $M_1$ and $M_2$). The machine adds this edge to the matching $M_2$. The large machine then sends to the sublinear machines all edges in the matching as described in Claim~\ref{consec_machines}. 

\paragraph*{Phase 3:}  
Let $E''$ be the set of edges that have both endpoints unmatched in $M_1 \cup M_2$.
Each small machine sends the large machine a count of the number
of edges in $E''$ that it stores, and the large machine sums these counts to compute $\vert E''\vert$. If $\vert E''\vert > 2n$,
the algorithm fails. Otherwise, the small machines send $E''$ to the large machine, and the large machine computes a maximal matching $M_3$ over $E''$.

The final output of the algorithm is $M_1 \cup M_2 \cup M_3$.
This is indeed a maximal matching of $G$, as in Phase~3 the large machine receives all edges whose two endpoints remain unmatched in $M_1 \cup M_2$, and $M_3$ completes $M_1 \cup M_2$ into a maximal matching.

To prove that the algorithm succeeds w.h.p., we need only to show that it does not fail in Phase 3, i.e., that the number of edges the small machines
need to send to the large machine is not too large:
\begin{lemma}
\label{lem:mm_phase_3_edge_count}
After Phase~2, w.h.p.\ the total number of edges whose two endpoints are unmatched in $M_1 \cup M_2$  is at most $2n$.
\end{lemma}
\begin{proof}[Proof.]
First,
since $M_1$ is a maximal matching over the subgraph induced by $V_{\ell}$,
no edge that has both endpoints in $V_{\ell}$ still has both endpoints unmatched in $M_1 \cup M_2$.
Thus, 
we consider only edges that have at least one endpoint in $V_h$.
Let $u_1,\ldots,u_k$ be the vertices of $V_h$, in the order they are processed by the large machine when constructing $M_2$,
and let $R_i$ be the randomness used to select the $\Theta(d\log{n})$ random neighbors of $u_i$ sent to the large machine.

If $\deg(u_i) < 2d\log{n}$, then no incident edges of $u_i$ have both endpoints
unmatched after $u_i$ is processed: in this case the large machine is sent all incident edges of $u_i$,
and when it processes $u_i$, if there is some edge $\sett{u_i, v} \in E$ such that $u_i, v$ are both unmatched,
the large machine adds some incident edge of $u_i$ to the matching.
Following this step, all incident edges of $u_i$ have at least one endpoint matched.

Suppose $\deg(u_i) \geq 2d\log{n}$,
and let $B_i$ be the event that
at the point where $u_i$ is processed by the large machine,
$u_i$ and
at least a $1/d$-fraction of $u_i$'s neighbors are still unmatched,
but the large machine is not able to find an edge $\sett{u_i, v} \in E'(u_i)$
such that $v$ is unmatched.
If $B_i$ does \emph{not} occur, then either at most $\deg(u_i) / d$ incident edges of $u_i$ have both endpoints unmatched after $u_i$ is processed, or the large machine has found an edge $\sett{u_i,v} \in E'(u_i)$ such that $v$ is unmatched.

Fix the randomness $R_1,\ldots,R_{i-1}$ of the nodes preceding $u_i$,
and assume that given this fixing, node $u_i$ is unmatched after processing and has at least $\deg(u_i) / d$ unmatched neighbors when it is processed (otherwise, the event $B_i$ does not occur). The randomness used to select 
$E'(u_i)$ is independent of $R_1,\ldots,R_{i-1}$.
Each time we select the next random neighbor of $u_i$ for $E'(u_i)$,
as long as we have not yet selected an unmatched neighbor,
the probability that we select an unmatched neighbor is at least
$(\deg(u_i) / d) / \deg(u_i) = 1/d$.
%(as we proceed, the number of unmatched neighbors does not decrease,
%but the total number of neighbors does decrease).
Thus, the probability that we fail to select any of $u_i$'s unmatched neighbors is at most $(1 - 1/d)^{2d\log{n}} \leq e^{-2\log{n}} < 1/n^2$.
Since this holds for every fixing of $R_1,\ldots,R_{i-1}$,
we have $\Pr[B_i] \leq 1/n^2$,
and by union bound,
the probability that none of the events $B_i$
for nodes $u_i$ with $\deg(u_i) \geq 2d\log{n}$
occur is at least $1 - 1/n$.
In this case,
the total number of edges that still have both endpoints unmatched
after the large machine finishes processing the nodes of $V_h$
is at most
$\sum_{i = 1}^k \deg(u_i)/d \leq 2nd/d = 2n$.
%Since $e$ is not 
%it is unmatched by $M_1$, which is a maximal matching on the induced graph on the vertices %$V_l$, at least of of $e$'s endpoints $u$ is contained in $V_h$.
%Therefore, in Phase~2 the small machines sample $d\log{n}$ random neighbors of $u$ and send them to the large machine. At the step in Phase~2 in which $u$ is processed, the $O(d\log{n})$ random edges of $u$ are considered by the large machine, and all of them have the other endpoint matched (otherwise $u$ is matched in $M_2$). Therefore, with high probability at most $1/d$ fraction of $u$'s neighbors are unmatched by $M_1 \cup M_2$. Moreover, since $m = \frac{nd}{2} = \frac{\sum_{v \in V} \deg(v)}{2} \geq \frac{\sum_{v \in V_h} \deg(v)}{2}$, we get that $\frac{\sum_{v \in V_h} \deg(v)}{d} \leq n$. 
\end{proof}

This concludes our algorithm. For the more general setting where the large machine has memory of size $\tilde{O}(n^{1+f(n)})$,
the MapReduce algorithm of~\cite{filtering} can be adapted to find a maximal matching as well:

\begin{theorem}\label{MM_with_super_lin}
Given a single machine with memory $\tilde{\Omega}(n^{1+f(n)})$, 
and $\Omega(m / n^{\gamma})$
machines with memory $\tilde{\Omega}(n^{\gamma})$,
there is an $O(1/f(n))$-round algorithm in the \modelname{} model that computes a maximal matching with high probability.
\end{theorem}

\begin{proof}[Proof of theorem~\ref{MM_with_super_lin}.]
We follow the algorithm of \cite{filtering} for finding a maximal matching the MapReduce model. By \cite{filtering} Lemma~3.1, if we sample each edge of an input graph $G$ independently with some probability $p$ and find a maximal matching $M$ over the sampled subgraph $G_p$, the number of edges in $G$ whose both end-points are unmatched in $M$ (i.e., potentially can be added to the matching) is $O(n/p)$ with high probability. 

The algorithm is recursive. If the current graph has $O(n^{1+f(n)})$ edges, we reach the stop condition, and send all edges to the large machine to find a maximal matching. Otherwise, we sample each edge independently with probability $p = \frac{1}{n^{f(n)}}$ to create a random subgraph $G_p$. When the algorithm returns, a maximal matching $M$ of the graph $G_p$ is known to the large machine. The large machine notifies the small machines about the matched vertices in $M$ as described in Claim~\ref{spreading_claim}. Denote the set of edges that can potentially be added to the matching $M$ by $E_M$.
From the statement in the beginning of the proof, we have that $\vert E_M\vert = O(n^{1+f(n)})$, thus they can fit the memory of the large machine. Finally, the large machine computes the matching $M_2$ such that $M \subseteq M_2$ and returns $M_2$ from the recursive call. The algorithm takes at most $O(1/f(n))$ recursive iterations, at which point it reaches the stop condition.

\end{proof}

%% file: Conclusion.tex
\section{Conclusion and Future Directions}

In this first work we focused on a heterogeneous MPC regime where we add a single machine with near-linear (or larger) memory to the sublinear regime. We showed that even a single machine can make a big difference in the round complexity of fundamental graph problems, circumventing several conditional hardness results. This is only one special case of the heterogeneous setting, and in general one can ask --- just how many machines of each memory size (sublinear, near-linear or superlinear) are required to efficiently solve a given problem? 
And if we allow the total memory use of all the machines to exceed the input size (as in, e.g.,~\cite{subgraph_counting,mapreduce,hardness}), does that help even further?
To study these questions, we propose a more general version of the heterogeneous model,
parameterized by the total memory of each type of machine:
the $(S_{\text{sub}}(m,n),S_\text{lin}(m,n),S_\text{sup}(m,n))$-\modelname{} model
has machines with sublinear, near-linear or super-linear memory,
using a total of $\tilde{\Theta}(S_{\text{sub}}(m,n)),\tilde{\Theta}(S_\text{lin}(m,n))$ and $\tilde{\Theta}(S_\text{sup}(m,n))$
memory, respectively.
(This means that the number of near-linear machines is $\tilde{\Theta}(S_\text{lin}(m,n)/n)$,
while the number of sublinear and near-linear machines 
is $\tilde{\Theta}(S_{\text{sub}}(m,n) / n^{\gamma})$ and $\tilde{\Theta}(S_{\text{sup}}(m,n) / n^{1+\gamma})$,
respectively, for some parameter $\gamma \in (0,1)$.)
From this more general perspective, the model that we studied in this paper
is the $(\tilde{\Theta}(m), \tilde{\Theta}(n), 0)$-\modelname{} model.

We mention several concrete open problems. First, it is interesting to ask whether MST and maximal matching can be solved in the \modelname{} model (the specific version that we studied in this paper) in the same round complexity as in the near-linear model,
and whether this holds for the other problems that
we did not address here --- e.g., those whose conditional hardness is proven in~\cite{hardness}.
Another question is whether problems that appear to require \emph{polynomial} time in the near-linear model, but can be efficiently solved in the superlinear model, can also
be efficiently solved in a hybrid near-linear / superlinear model where we have many near-linear machines and a small number of superlinear machines.
%\Anote{is this last sentence still interesting after the first paragraph we add?}
%\Rnote{Sure, why not, it's a specific instance. Also we need to fill up this space :-)}
Finally, it is very interesting to ask whether any conditional hardness results can be strengthened to yield non-trivial lower bounds for the heterogeneous model, in order to better characterize the benefit resulting from adding one near-linear machine to the sublinear MPC model.

%% file: app_spanners.tex
\section{Constructing the Clustering Graphs from Section~\ref{sec:spanners}}
\label{app:spanners}

We provide a more detailed description of the clustering graphs from \cite{spanners}.

A \emph{star} in a graph $G$ is a tree-shaped subgraph of $G$ with one or more vertices and diameter at most $2$.

For graph $G = (V,E)$ of maximum degree $\Delta$, let $S_0,...,S_{\log\Delta-1}$ be sets of stars in $G$. From the \emph{star} definition, each $s \in S_i$ is either a single vertex $u \in V$, or a center vertex $u \in V$ surrounded by a set of vertices connected to it. In both cases, we call $u$ the center of the star $s$, and denote it by $c(s) = u$. Moreover, for $v \in V$ denote by $s_i(v)$ the center of star to which $v$ belongs in the set $S_i$, i.e., if $v$ belongs to the star $s \in S_i$ then $s_i(v) = c(s)$ (in case $v$ does not belong to any of the stars in $S_i$, $s_i(v)$ is undefined).

Now, define the corresponding clustering graphs $A_0,...,A_{\log\Delta-1}$ of the sets $S_0,...,S_{\log\Delta-1}$ in the following way: $V_i = \sett{c(s) \mid s \in S_i}$, and 
\begin{align*}
                 E_i &= \{(c,c') \mid \exists(u,v) \in E: s_i(u)=c , s_i(v)=c' ,
                 \\
                 &  \min\sett{\deg(u), \deg(v)} \in [2^i,2^{i+1}) 
                   \}
\end{align*}

Given an edge $(c,c') \in E_i$, let $E_G((c,c')) \in E$ be the lexicographically-smallest edge $(u,v) \in E$ such that $s_i(u)=c$, $s_i(v) = c'$ and $\min\sett{\deg(u), \deg(v)} \in [2^i,2^{i+1})$.
(such an edge must exist, from the definition of $E_i$.)
Similarly, for a set of edges $X \subseteq E_i$, we let $E_G(X) = \bigcup_{ (c,c') \in X} E_G((c,c'))$.

\begin{lemma}[\cite{spanners} Section 4]\label{lem:clustering}
        For a graph $G=(V,E)$ of maximum degree $\Delta$, there exist star sets $S_0,...,S_{\log\Delta-1}$ and corresponding clustering graphs $A_i=(V_i,E_i)$ such that:
\begin{itemize}
    \item $ \vert V_0 \vert  = n,  \vert E_0 \vert  = O(n)$, and for each $0 \leq i \leq \log\Delta-1$, $ \vert V_i \vert  = O(ni/2^i)$ and $ \vert E_i \vert  = O(n\cdot 2^i)$.
    \item For any $e \in E$, there exists $0 \leq i \leq \log\Delta-1$ such that $e$ is either contained in some star $s \in S_i$, or $e \in E_i$.
        \item If a vertex $v \in V$ belongs to two stars $s \in S_i$ and $s' \in S_j$ for $i \neq j$ such that $v \not = c(s)$ and $v \not = c(s')$, then $c(s) = c(s')$.
\end{itemize}
\end{lemma}

Let $A_0,...,A_{\log\Delta-1}$ be the clustering graphs from Lemma~\ref{lem:clustering}. 
The following lemma states that by combining individual spanners for these clustering graphs we can obtain a spanner for the original graph:
\begin{lemma}[\cite{spanners} Section 4.4]\label{lem:combining_spanners}
    Let $H_0,\dots,H_{\log\Delta-1}$ be $(2k-1)$-spanners of $A_0,\dots,A_{\log\Delta-1}$, respectively. Let $E_{stars}^i = \cup_{s\in S_i}\sett{e \mid e \in s}$ for all $0 \leq i \leq \log\Delta-1$. 
    Then, $H = \bigcup_{0 \leq i \leq \log\Delta-1} E_G(H_i) \cup E_{stars}^i$ is a $(6k-1)$-spanner of $G$,
    of size at most $O(n) + \sum_{i = 0}^{\log\Delta - 1}  \vert H_i \vert $.
\end{lemma}

\paragraph*{Clustering graphs in the \modelname{}.} 
Algorithm~\ref{alg:clustering} constructs the clustering graphs $A_0,...,A_{\log\Delta-1}$ from Lemma~\ref{lem:clustering} in $O(1)$ rounds in the \modelname{} model in the same manner as \cite{spanners}. We define the sets $U_i = \{v \in V \mid 2^i \leq deg(v) \leq 2^{i+1} \}$ for $0 \leq i \leq \log\Delta$. The set $B_i \subseteq V$ is called a hitting-set of $U_i$ if for each $v \in U_i$, either $v \in B_i$ or $v$ has a neighbor $u \in B_i$. Algorithm~\ref{alg:clustering} starts by computing the hitting-set $B_i$ for each $U_i$. Then, the clustering graphs $A_i = (V_i,E_i)$ for each $0 \leq i \leq \log\Delta-1$ are constructed based on these hitting-sets. We note that for each $e=(c,c') \in E_i$, we attach the lexicographically-smallest edge $(u,v) \in E$ that cause $e$ to be contained in $E_i$, so that in the end, for each $(c,c') \in H_i$ we will be able to compute $E_G((c,c'))$ and add it to the final spanner $H$ of $G$.

%% file: app_MST.tex
\section{Pseudocode}
\label{app:MST}
In this section, we provide pseudo-code for key procedures of Section~\ref{sec:MST} in the \modelname{} model, namely, the doubly-exponential Bor{\r u}vka procedure, and a procedure for identifying the $F$-light edges.

We also provide the pseudo-code for Section~\ref{sec:spanners}, for constructing the clustering graphs and for constructing the spanner.

\begin{algorithm2e*}
\caption{DoublyExpBoruvka($G=(V,E,w),k$)}\label{alg:lotker}
\DontPrintSemicolon
\SetKwInOut{Input}{Input}
\SetKwInOut{Output}{Output}
\SetKw{Disseminate}{disseminate}
\SetKw{Sort}{sort}
\SetKw{Aggregate}{aggregate}
\SetKwBlock{SmallMachines}{small machines do in parallel:}{end}
\SetKwBlock{LargeMachine}{large machine:}{end}
\SetKwBlock{LocallyCompute}{locally compute:}{end}
\Input{$G$ is a weighted undirected graph, k is an integer}
\Output{$G_k$ is the contracted graph resulting from applying $k$ iterations of $\Boruvka$, $F$ are the edges that were used for contraction}
\hrulealg
$F \gets \phi, V_0 \gets V, E_0 \gets E, \forall v \in V: c_{0}(v)=v$ \tcp*{$c_i : V \rightarrow V_i$ map each vertex $v \in V$ to the vertex into which it was merged in $V_i$}
\For{$0 \leq i \leq k-1$} {
\SmallMachines{
\Sort all edges $E_i$ primarily by $1^{st}$ endpoint id, then by weight, then by the $2^{nd}$ endpoint id \tcp*{Claim~\ref{sorting_claim}}
}
\LargeMachine{
\LocallyCompute{
\ForEach{$v \in V$}
{compute $deg_{out}(v)$ and $M_{first}(v)$ \tcp*{Claim~\ref{consec_machines}}
\ForEach{small machine $M$}
{compute $k(v,M) \gets $ out of $v$'s lightest $\min(2^{2^i}, \deg_{\mathrm{out}}(v))$ edges, the number of edges that are stored on $M$}
}}
\Disseminate $k(v,M)$ to small machine $M$ that holds an adjacent edge of $v$ for each $v \in V$ \;}
\SmallMachines{
\ForEach{$v \in V$ and machine $M$}
{$E_{lightest} \gets $the $k(v,M)$ lightest edges of $v$ that are stored on $M$}
send $E_{lightest}$ to the large machine\;
}
\LargeMachine{
\LocallyCompute{
$F_{contracted} \gets \phi$, $V_{i+1} \gets \phi$\;

\ForEach{$e \in E_{lightest}$ from the lightest to the heaviest} {
\If{$e = (u,v)$ dose not create a cycle in $F_{contracted}$} {
$F_{contracted} \gets F_{contracted} \cup \{(u,v)\}$\;
$V_{i+1} \gets V_{i+1} \cup \sett{\min\sett{u,v}}$\;
$c_{i}^{'}(u),c_{i}^{'}(v) \gets \min\{u,v\}$ \tcp*{$c'_i : V_i \rightarrow V_{i+1}$ map each $v \in V_i$ to the vertex into which it was merged in $V_{i+1}$}
$F \gets F \ \cup \{ (u',v') \in E \mid c_i(u') = u, c_i(v') = v, (u',v')$ is the minimum weight such edge\}\;
}
}
}
\Disseminate $c_i^{'}(v)$ to each machine that holds an adjacent edge of $v$ for each $v \in V$
\tcp*{Claim~\ref{spreading_claim}}
}
\SmallMachines{
\Aggregate $E_{i+1} \gets \{(u'',v'',w(u,v)) \mid (u,v) \in E_i, c_i^{'}(u) = u'', c_i^{'}(v) = v''$, $u'' \not = v''$ and $(u,v)$ is the minimum weight edge for which the previous conditions hold\} \tcp*{Claim~\ref{aggregation_claim}}
\ForEach{$v \in V$} {
$c_{i+1}(v) = c_i'(c_i(v))$\;
}
}
}

output $(G_k = (V_k, E_k, w), F)$\;
\end{algorithm2e*}

\begin{algorithm2e*}
\caption{F-LightEdges($G=(V,E,w),F$)}\label{alg:F_light}
\DontPrintSemicolon

\Input{$G$ is a weighted, undirected graph, $F$ is a spanning forest of $G$ stored on the large machine}
\Output{$E_{light}$ are the $F$-light edges in $G$}\tcp*{see section \ref{sec:MST} for $F$-light, $\mathcal{M}_{\mathrm{flow}}$ and $\mathcal{D}_{\mathrm{flow}}$ definitions}
\hrulealg

$E_{light} \gets \phi$\;
\LargeMachine {
\LocallyCompute{
compute labels $L:V \rightarrow \sett{0,1}^{O(\log^2 n)}$ by applying $\mathcal{M}_{\mathrm{flow}}(F)$\;
}
\Disseminate $L(v)$ to each small machine $M$ that holds an edge incident to $v$ for each $v \in V$ \tcp*{Claim~\ref{spreading_claim}}
}

\SmallMachines{
\ForEach{$(u,v) \in E$} {
\If{$w(u,v) \leq \mathcal{D}_{\mathrm{flow}}( L(u), L(v) )$} {
$E_{light} \gets E_{light} \cup \sett{(u,v)}$
}
}
}
output $E_{light}$
\end{algorithm2e*}

\clearpage

%% file: app_spanner2.tex
% \section{Pseudocode from Section~\ref{sec:spanners}}

\begin{algorithm2e}
\caption{ClusteringGraphs(G)}\label{alg:clustering}
\DontPrintSemicolon

\Input{$G$ is a weighted, undirected graph}
\Output{$H$ are the edges used to construct the clustering graphs, $A_0,...,A_{\log\Delta-1}$ are the clustering graphs}
\hrulealg

\LargeMachine{
\LocallyCompute{
\For{$1 \leq j \leq \log(n)$} {
$D_0^j \gets V$\;
\For{$1 \leq i \leq \log\Delta-1$} {
$D_i^j \gets$ sample each $u \in V$  w.p. $p = \frac{i}{2^i}$\;
}
}
}
\Disseminate $D_i^j$ to each machine that holds an incident edge of $v \in D_i^j$ for each $v \in V$, $1 \leq j \leq \log(n)$ and $1 \leq i \leq \log\Delta-1$  \tcp*{Claim~\ref{spreading_claim}}
}

\SmallMachines{
\DoParallel{$1 \leq j \leq \log(n)$} {
\DoParallel{$1 \leq i \leq \log\Delta-1$} {
\Aggregate $D_i^j \gets D_i^j \cup \{u \in V \mid 2^k \leq deg(u) \leq 2^{k+1}, i \leq k, \ u $ has no neighbor in $D_i^j\}$ \tcp*{Claim~\ref{aggregation_claim}}
send $D_i^j$ to the large machine\;
}
}
}
\LargeMachine{
\LocallyCompute{
\For{$1 \leq i \leq \log\Delta-1$} {
$D_i \gets D_i^j$ such that 
$D_i^j = \argmin_{1 \leq j \leq \log{n}}\{\vert D_i^j \vert\}$\;
}
}
\Disseminate $D_i$ to each machine that holds an incident edge of $v \in D_i$ for each $v \in V$ and $1 \leq i \leq \log\Delta-1$ \tcp*{Claim~\ref{spreading_claim}}}
\SmallMachines{
\DoParallel{$0 \leq i \leq \log\Delta-1$} {
$B_i \gets \cup_{j=i}^{\log\Delta - 1}D_j$\;
}
\ForEach{$u \in V$} {
\Aggregate $i_u \gets \max\{i \mid u \in B_{i}$ or $N(u) \cap B_{i} \not = \emptyset\}$ \tcp*{Claim~\ref{aggregation_claim}}
\If{$u \in B_{i_u}$} {
$\sigma_u \gets u$\;
\DoParallel{$0 \leq i \leq i_u$} {
$V_i \gets V_i \cup \{u\}$\;
}
}
\Else {
$\sigma_u \gets $ a random neighbor of $u$ in $B_{i_u}$ \;
$H \gets (u,\sigma_u)$\;
}
}
\ForEach{$(u,v) \in E$} {
\If{$\min\{\deg(u),\deg(v)\} = 2^i$ and $\sigma_u \not = \sigma_v$} {
$E_i \gets (\sigma_u, \sigma_v)$\;
}
}
}
output $H$ and $(V_i,E_i)$ for $0\leq i\leq \log\Delta-1$\;

\end{algorithm2e}

\begin{algorithm2e}
\caption{Spanner($G,k$)}\label{alg:spanner_alg}
\DontPrintSemicolon

\Input{$G$ is a weighted, undirected graph, $k$ is an integer}
\Output{$H$ is an $O(k)$-spanner of $G$ of expected size $O(n^{1+\frac{1}{k}})$}
\hrulealg

$H,A_0, ... ,A_{\log\Delta-1} \gets$ ClusteringGraphs($G$) \tcp*{$A_i = (V_i, E_i)$}

\DoParallel{$0\leq i \leq \log\Delta-1$} { 
\If{ $i=0$ or $p_i = \frac{k^2i^{1+1/k}}{2^i} > 1$} { 
\SmallMachines{
send $A_i$ to the large machine\;
}
\LargeMachine{
\LocallyCompute{
$H_i \gets (2k-1)$-spanner of $A_i$\;
$H \gets H \cup E_G(H_i)$\; \tcp*{$E_{G}(H_i) = \{(u,v) \in E \  \mid  \min\{\deg(u),\deg(v)\} = 2^i, \ (\sigma_u, \sigma_v) \in H_i$ and $(u,v)$ is the minimum weight such edge$\}$}
}
}
}
\Else {
\SmallMachines{
\DoParallel{$1 \leq j \leq k-1$} {
$E_i^j \gets$ sample each $e \in E_i$ w.p. $p_i=\frac{k^2i^{1+1/k}}{2^i} > 1$\;
send $A_i^j=(V_i,E_i^j)$ to the large machine\;
}
}
\LargeMachine{
\LocallyCompute{
$H_i \gets$  compute lines~\ref{line:BS_start}--\ref{line:BS_mid}
of ModifiedBaswanaSen($A_i,k,1/2^i$) using the sampled subgraphs $(A_i^1,...,A_i^{k-1})$\;
\ForEach{$v \in V_i$} {
$l_i^v \gets (i,c_0(v),...,c_{t-1}(v))$ such that $v$ stopped being clustered at iteration $t $, $c_j(v)$ is $v$'s cluster center at iteration $j$\;
}
}
\Disseminate $l_i^v$ to each machine that holds edges of $v$ for each $v \in V$\tcp*{Claim~\ref{spreading_claim}}
}
\SmallMachines{
\ForEach{$(u,v) \in E_i$}{
\If{$ \vert l_i^v \vert  >  \vert l_i^u \vert $} {
Create a record $(i,u,c_{t-1}(v),v)$\;
}\tcp*{$u$ stopped being clustered before $v$ at iteration $t$}
\ElseIf{$ \vert l_i^u \vert  >  \vert l_i^v \vert $} { 
Create a record $(i,v,c_{t-1}(u),u)$\;
}\tcp*{$v$ stopped being clustered before $u$ at iteration $t$}
}
\ForEach{$u \in V_i$ and a cluster $c$} {
\Aggregate $v \gets \argmin_{v \in V_i}\{ (i,u,c_{t-1}(v),v) \mid c_{t-1}(v) = c\}$ \tcp*{Claim~\ref{aggregation_claim}}
$H_i \gets H_i \cup \sett{(u,v)}$\;
}
$H \gets E_G(H_i)$\;
}
}
}
output $H$
\end{algorithm2e}

\clearpage

%% file: app_additional_results.tex
\section{Prior Works That Extend to the \modelname{} Model}
Several prior works in the near-linear MPC model, while not explicitly presented for a model like \modelname{}, easily translate to this model. In this section, we give an overview of these results, and provide some relevant adaptation details.

\subsection{Connectivity in $O(1)$ Rounds}\label{app:conn}

In this subsection we show how by leveraging the existing techniques of linear sketches and $l_0$-sampling which are mostly used in the context of streaming algorithms \cite{l0_sampling,linearMeas,linear_sket_stream_1,linear_sket_stream_2,linear_sket_stream_3,linear_sket_stream_4}, as well as in the congested clique (\cite{connectivity_in_cc}, we get constant-round algorithms for connectivity in the \modelname{} model.
The implementation is very similar to the implementation in near-linear MPC, except that now the edges associated with a single vertex may be stored across more than a single machine;
using \emph{linear sketches}, this is trivial to overcome.

\begin{theorem} \label{conn_theorem}
There is an $O(1)$-round algorithm to identify the connected components of a graph with high probability in the \modelname{}.
\end{theorem}

\begin{proof}[Proof.]
We apply the connectivity algorithm from ~\cite{linearMeas} which is based on linear-sketches for $\ell_0$-sampling, as stated originally in \cite{l0_sampling},
but replacing the shared randomness assumed in~\cite{l0_sampling}
with $O(\log n)$-wise independence~\cite{linearMeas},
so that one machine can generate $O(\polylog n)$ random bits,
disseminate them to all the other machines,
and these are then used to generate all the sketches.
The sublinear machines compute a linear sketch $s(v)$ for each node $v \in V$. This sketch can be used to sample a random neighbor of $v$, and for a set of nodes $S = \{v_1,...,v_k\}$, the sketch $s = s(v_1) + ... + s(v_k)$ can be used to randomly sample an edge from $E[S,V \backslash S]$. The total number of bits required for all nodes sketches is $O(n\log^3(n))$, thus, all sketches can be stored in the near-linear-spaced machine which can then simulate the connectivity algorithm of \cite{linearMeas}. During the algorithm, nodes from the same connected component are merged using the sampled edges into contracted virtual nodes. The nodes in the final contracted graph represent the connected components of the input graph.  

In order to construct a sketch for some node $v$, it is required to know all neighbors of $v$. But, we cannot assume that the sublinear machine has a sufficient memory to hold all of $v$'s neighbors. Instead, we use the following property of the linear sketches:

\begin{property} \label{linearity_property}
Let $G_1,\dots,G_k \subseteq G$ such that $N(v) = N_{G_1}(v) \uplus \dots \uplus N_{G_k}(v)$, and let $s_1(v),\dots,s_k(v)$ be linear sketches in these graphs, then $s(v) = \sum_{i=1}^k s_i(v)$ is a linear sketch of $v$.
\end{property}

Each sublinear machine computes a partial-sketch for each $v \in V$ based on the neighbors of $v$ which it holds. Using Claim~\ref{aggregation_claim}, the final sketches are constructed by adding up all the partial sketches, which takes constant number of rounds.
\end{proof}

\subsubsection{$(1+\epsilon)$-Approximation of MST in $O(1)$ Rounds}

\begin{theorem}
For any constant $\epsilon > 0$, there is an $O(1)$-round algorithm to compute an $(1+\epsilon)$-approximation of the minimum spanning tree in the \modelname{}.
\end{theorem}

\begin{proof}[Proof.]
We follow the known reduction from \cite{linearMeas} to reduce the problem of estimating the weight of the minimum spanning tree to the problem of counting the number of connected components in the graph. This approach requires applying the connectivity algorithm in parallel over $O(\log(n))$ subgraphs. This takes constant number of rounds using the connectivity algorithm from theorem \ref{conn_theorem}.
\end{proof}

\subsection{Exact Unweighted Minimum Cut in $O(1)$ Rounds}
In this subsection, we give an overview of the result of \cite{unweightewMincut} for computing the minimum cut of an unweighted graph.

\begin{theorem}
There is an $O(1)$-round algorithm in the \modelname{} model that with high probability computes the minimum cut of an unweighted graph.
\end{theorem}

\begin{proof}[Proof.]
The heart of the algorithm is a contraction process, which consists of two parts: 
\begin{enumerate}
    \item \textit{$2$-out contraction}: Each node $v \in V$ samples two edges. Then, we contract all connected components of the graph induced over the sampled edges, while allowing parallel edges.
    \item \textit{Random-sampling contraction}: For a multi-graph of minimum degree $\delta$, we contract a set of edges $E_p$, to which we include each $e \in E$ independently with probability $p = \frac{1}{2\delta}$. 
\end{enumerate}

Combining the $2$-out contraction with random-sampling contraction results in a graph $G_c$ with $\vert E_c\vert = O(n)$ and $\vert V_c\vert = O(n/\delta)$, such that, for a constant $\epsilon \in (0,1]$ and any non-singleton $(2-\epsilon)$-minimum-cut, the cut is preserved in $G_c$ with at least a constant probability. We can amplify the success probability in the same way as in \cite{unweightewMincut}, so that with high probability, all non-trivial minimum cuts are preserved. Since $\vert E_c\vert = O(n)$, $G_c$ can fit the memory of the large machine, which can then compute its minimum cut and compare it with all the singleton cut sizes known to it by Claim~\ref{consec_machines}. The whole process takes $O(1)$ rounds in the \modelname{} model.

% For the \textit{$2$-out contraction}, the small machines sample two edges for each $v\in V$ and send to the large machine. The large machine contracts the sampled graph and notifies the small machines about the contracted nodes using Claim~\ref{spreading_claim}. For the \textit{Random-sampling contraction}, each small machine that holds an edge $e \in E$ samples it with probability $p$. All the sampled edges are sent to the large machine for contraction, and then the large machine updates about the contractions again using Claim~\ref{spreading_claim}.

\end{proof}

\subsection{$(1 + \epsilon)$-Approximating the Minimum Cut in $O(1)$ Rounds}
In this subsection, we overview a procedure of \cite{weightewMincut} which implies an $O(1)$-round algorithm in the \modelname{} model for $(1 + \epsilon)$-approximation of the minimum cut.

\begin{theorem}\label{approx-mincut}
There is an $O(1)$-round algorithm in the \modelname{} model that $(1 + \epsilon)$-approximate the minimum cut of a weighted graph with high probability.
\end{theorem}

\begin{proof}[Proof.]
Algorithm 1 of \cite{weightewMincut} is an algorithm that given a weighted graph $G$ with minimum cut $\lambda^*$, w.h.p computes an $O(n\log n)$ edges unweighted multi-graph $G'$, such that all preserved cuts of $G$ in $G'$ have sizes within $1 \pm \epsilon$ of their expectation, and a minimum cut of $G$ is preserved in $G'$ with a constant probability. %They transform the graph into a multi-graph of unweighted edges, and % 
The algorithm uses two sub-procedures. The first, reduces the number of edges to $O(n\lambda^*)$ while preserving the minimum cut value. The second, is a sampling procedure which takes each edge with probability $c\log{n}/(\epsilon^2\lambda)$ for a value $\lambda$ close to $\lambda^*$, and reduces the size of the graph to $O(n\log{n})$ edges while preserving the cuts within a $(1 \pm \epsilon)$ of their expectations (assuming the min-cut is $\Omega(\log{n})$. In case it is smaller, the graph from the first procedure already has $O(n\log{n})$ edges, so we skip this step). 

\begin{lemma}[Lemma 2.4 and Remark 2.5 from \cite{weightewMincut}, Informal]
% Given a weighted graph $G$ with minimum cut $\lambda^*$ to compute an $O(n\log{n})$ edge unweighted multi-graph $G'$. With high probability, all preserved cuts of $G$ in $G'$ have sizes within a $(1 \pm \epsilon)$ of their expectations. A minimum cut of $G$ is preserved with some constant probability. 
The complexity of this algorithm in the MPC model depends only on the complexity of the Connected Components algorithm in that model.
\end{lemma}

In the \modelname{} model, the Connected Components problem can be solved in $O(1)$ rounds, as shown in Section~\ref{sec:MST}. Thus, the algorithm from \cite{weightewMincut} can be implemented in $O(1)$ rounds. The small machines then send the multi-graph $G'$ to the large machine, which outputs the minimum cut, re-normalized to its original value.
\end{proof}

%It is stated in \cite{weightewMincut} Remark, that the complexity of this algorithm in the sublinear MPC model depends only on the complexity of the Connected Components algorithm:
%\Anote{this is a citation:}
%\textit{
%\begin{enumerate}
%    \item The operations of the first kind are operations that are applied
%to the edges locally, like verifying whether its weight is larger
%than $c\lambda$ or edge sampling
%\item The second kind of operation contracts a set of edges that
%are selected in some way (either at random, or by weight).
%To contract a set of edges, we only need to solve a Connected
%Components problem, hence the round and memory complexity
%matches that of the Connected Components algorithm.
%\end{enumerate}}

% \begin{equation*}
%     (1-\epsilon)p\lambda^{''} \leq (1-\epsilon)\lambda^{''}_{Exp} \leq \lambda_p \leq \lambda' \leq (1+\epsilon)\lambda'_{Exp} = (1+\epsilon)p\lambda^*
% \end{equation*}

\subsection{MIS in $O(\log\log(\Delta))$ Rounds}
In this subsection, we overview the result \cite{MIS}, which naturally extends to the \modelname{} model, and give some basic details on this algorithm.

\begin{theorem}
There is an $O(\log\log\Delta)$-round algorithm in the \modelname{} model that with high probability computes a Maximal Independent Set for a graph of maximum degree $\Delta$.
\end{theorem}

\begin{proof}[Proof.]
We follow the algorithm by \cite{MIS} and show that it can be implemented in the \modelname{} model to provide the same round complexity. 
The algorithm starts by having the large machine choosing a random permutation $\pi :[n] \rightarrow [n]$ and disseminating the ranks of the vertices to the small machines such that each small machine knows the rank of the end-points of the edges it stores using Claim~\ref{spreading_claim}. Let $\alpha = 3/4$. Repeat the following process iteratively for $i=0,1,\dots$ until the maximum degree is $O(\polylog{n})$:  
in iteration $i$, let $G_i$ be the graph induced by the vertices of rank $n/\Delta^{\alpha^i}$ to $n/\Delta^{\alpha^{i+1}}$. It can be shown inductively that $G_i$ has  $\tilde{O}(n)$ edges. The small machines send the edges of $G_i$ to the large machine, which then locally applies the following algorithm on $G_i$:
\begin{itemize}
    \item Add the vertex $v$ that is next in the order defined by the permutation $\pi$ to the MIS.
    \item Remove all neighbors of $v$ from $G_i$.
\end{itemize}

In the end of the iteration, for each vertex $v\in V$ that was added to the MIS, the large machine notifies all the small machines that store an edge adjacent to $v$ that $v$ was taken into the MIS. These edges are then removed from the graph. It is shown that $O(\log\log \Delta)$ iterations suffice for reducing the number of remaining edges to $O(n\polylog{n})$.

\end{proof}

\subsection{$(\Delta+1)$-Vertex Coloring in $O(1)$ Rounds}
In this subsection, we implement the algorithm of \cite{coloring} into the \modelname{} to obtain an $O(1)$-round algorithm for $(\Delta+1)$-vertex coloring.

\begin{theorem}
There is an $O(1)$-round algorithm in the \modelname{} model that with high probability computes a ($\Delta+1)$-Vertex Coloring.
\end{theorem}

\begin{proof}[Proof.]

We use the following color-sampling lemma of~\cite{coloring}:

\begin{lemma}[Theorem 3.1 in \cite{coloring}]\label{lem:app_coloring_sample}
Let $G = (V, E)$ be any $n$-vertex graph and $\Delta$ be the maximum degree in G. Suppose for each vertex
$v \in V$ , we independently pick a set $L(v)$ of colors of size $\Theta(\log{n})$ uniformly at random from $\{0,\dots,\Delta\}$. Then, with high probability there exists a proper coloring $C : V \rightarrow \{0,\dots,\Delta\}$ of $G$ such that for all vertices $v \in V$, $C(v) \in L(v)$.
\end{lemma}

The algorithm works as follows: the large machine chooses $\Theta(\log{n})$ colors from $\{0,\dots,\Delta\}$ for each vertex, and disseminates them to the small machines using Claim~\ref{spreading_claim}, such that every small machine knows the colors of the vertices that participate in the edges it stores. For an edge $e = (u,v)$, we say that $e$ is conflicting if $L(u) \cap L(v) \neq \emptyset$. From Lemma 4.1 of \cite{coloring}, w.h.p. we have that the total number of conflicting edges in the graph is $O(n\polylog{n})$. The small machines send to the large machine all their hazardous edges. By Lemma~\ref{lem:app_coloring_sample}, the large machine may assign colors such that all conflicting edges are non-monochromatic. This implies that this coloring is proper for the entire graph, as only conflicting edges may be monochromatic in any such coloring.
\end{proof}